\newif\ifdraft
\newif\iffullversion

\fullversiontrue

\documentclass[11pt]{article}
\usepackage[T1]{fontenc}
\usepackage[latin1]{inputenc}
\usepackage{amsmath,amssymb,amsfonts,bbm,paralist,xspace,pst-node,multicol,microtype,marginnote,version,relsize}
\usepackage[letterpaper,dvips,
\ifdraft hratio=1:2\else
vmargin={1.01in,1.5in},hmargin=1in\fi]{geometry}
\usepackage[toc]{multitoc}
\usepackage[nottoc]{tocbibind}
\usepackage{makeidx}\makeindex
\iffullversion\usepackage{ae,aecompl}\fi
\usepackage[breaklinks,pdfauthor={Dominique Unruh},pdftitle={Universally Composable Quantum Multi-Party Computation}]{hyperref}
\usepackage[envcountsame,breakurls]{latexhacks}

\renewcommand\paragraph[1]{\medskip\noindent\textbf{#1}}

\iffullversion
\includeversion{fullversion}
\excludeversion{shortversion}
\newcommand\fullshort[2]{#1}
\else
\includeversion{shortversion}
\excludeversion{fullversion}
\newcommand\fullshort[2]{#2}
\fi
\newcommand\fullonly[1]{\fullshort{#1}{}}
\newcommand\shortonly[1]{\fullshort{}{#1}}

\AtBeginDocument{
  \ifx\PreviewMacro\undefined\else
  \PreviewEnvironment*{figure}
  \PreviewMacro*\footnote
  \PreviewMacro*\framebox
  \PreviewMacro*\parbox
  \PreviewMacro*\delaytext
  \PreviewMacro*\usedelayedtext
  \fi}

\ifx\definition\undefined
\newtheorem{definition}{Definition}
\newtheorem{theorem}{Theorem}
\newtheorem{corollary}{Corollary}
\newtheorem{lemma}{Lemma}
\newtheorem{conjecture}{Conjecture}
\newenvironment{proof}[1][]{\medskip\noindent
  \textit{Proof\ifx\proof#1\proof\else\ #1\fi. }}{}
\newcommand\qed{\hfill$\Box$}
\fi

\ifdraft
\AtBeginDocument{
\let\oldIndex\index
\DeclareRobustCommand\index[1]{\oldIndex{#1}\marginparx{\normalfont\textnormal\raggedright\tiny$\bullet$#1}}
\let\oldLabel\label
\DeclareRobustCommand\label[1]{\oldLabel{#1}\marginparx{\normalfont\raggedright\footnotesize\underline{#1}}}
}
\fi

\makeatletter

\newcommand\unruh@readlineopt{\let\unruh@temp\undefined\ifeof\unruh@tempread\else\immediate\read\unruh@tempread to\unruh@temp\fi}
\newread\unruh@tempread
\def\SVNrevision{-1}%
\IfFileExists{.svn/entries}{%
\immediate\openin\unruh@tempread .svn/entries\relax
\unruh@readlineopt
\unruh@readlineopt
\unruh@readlineopt
\unruh@readlineopt
\ifx\unruh@temp\undefined\else
\edef\SVNrevision{\number\unruh@temp}\fi
}{}

\newcommand\delaytext[2]{%
  \long\expandafter\gdef\csname delaytext:#1\endcsname{#2}%
  \AtEndDocument{%
    \expandafter\ifx\csname delaytext-used:#1\endcsname\relax
    \@latex@error{delaytext #1 unused. Use \noexpand\usedelayedtext{#1}
      somewhere}\fi}%
}
\newcommand\usedelayedtext[1]{%
  \expandafter\ifx\csname delaytext-used:#1\endcsname\relax\else
  \@latex@error{delaytext #1 used twice}\fi
  \global\@namedef{delaytext-used:#1}{1}%
  \expandafter\ifx\csname delaytext:#1\endcsname\relax
  \@latex@error{delaytext #1 undefined}\fi
  \@nameuse{delaytext:#1}%
}

\input{macros}

\newcommand\figureCorruptedAlice{
  \begin{pspicture}(-4.5,-5.8)(9,.9)
    \psset{nodesep=2pt}
    \rput(-4.5,.5){(a)}\rput(-3,0){\CAReal}
    \rput(-4.5,-2){(b)}\rput(-3,-2.5){\CAIdeal}
    \rput(-4.5,-4.8){(c)}\rput(-3,-4.8){\CAGameSix}
  \end{pspicture}}

\newcommand\CAReal{%
  \rput(0,0){\rnode{Adv}{$\Adv$}}
  \rput(2,0){\rnode{AC}{$A^C$}}
  \rput(4,-.3){\rnode{FCOM}{$\FCOM$}}
  \rput(6,0){\rnode{B}{$B$}}
  \rput(8,0){\rnode{Z}{$\calZ$}}
  \ncline[offset=2pt]{<->}{Adv}{AC}
  \ncline[offset=-3pt]{<-}{Adv}{AC}
  \ncarc{<->}{AC}{B}
  \ncarc{->}{B}{FCOM}
  \ncarc{->}{FCOM}{AC}
  \ncline[offset=2pt]{<-}{B}{Z}\Aput[2pt]{$\scriptstyle c$}
  \ncline[offset=-3pt]{->}{B}{Z}\Bput[2pt]{$\scriptstyle s$}
  \ncloop[angleB=180,linearc=.2]{<->}{Z}{Adv}
}

\newcommand\CAGameSix{%
  \rput(0,0){\rnode{Adv}{$\Adv$}}
  \rput(1.5,0){\rnode{AC}{$A^C$}}
  \rput(3.5,-.3){\rnode{FCOM}{$\FFC$}}
  \rput(5.5,0){\rnode{B}{$B$}}
  \rput(7,0){\rnode{AC2}{$A^C$}}
  \rput(8.5,0){\rnode{FROT}{$\FROT$}}
  \rput(10,0){\rnode{B2}{$\TIlde B$}}
  \rput(11.5,0){\rnode{Z}{$\calZ$}}
  \ncline[offset=2pt]{<->}{Adv}{AC}
  \ncline[offset=-3pt]{<-}{Adv}{AC}
  \ncarc{<->}{AC}{B}
  \ncarc{->}{B}{FCOM}
  \ncarc{->}{FCOM}{AC}
  \ncline{->}{B}{AC2}\Aput[2pt]{$\scriptstyle s_0,s_1$}
  \ncline{->}{AC2}{FROT}\Aput[2pt]{$\scriptstyle s_0,s_1$}
  \ncline[offset=2pt]{<-}{FROT}{B2}\Aput[2pt]{$\scriptstyle c$}
  \ncline[offset=-3pt]{->}{FROT}{B2}\Bput[2pt]{$\scriptstyle s_c$}
  \ncline[offset=2pt]{<-}{B2}{Z}\Aput[2pt]{$\scriptstyle c$}
  \ncline[offset=-3pt]{->}{B2}{Z}\Bput[2pt]{$\scriptstyle s_c$}
  \ncloop[angleB=180,linearc=.2]{<->}{Z}{Adv}
  \psframe[linestyle=dashed,framearc=.2](-.7,.5)(5.8,-.7)
  \rput(-.4,-.9){$\Sim$}
}

\newcommand\CAIdeal{%
  \rput(0,0){\rnode{Adv}{$\Sim$}}
  \rput(2,0){\rnode{AC}{$A^C$}}
  \rput(4,0){\rnode{FROT}{$\FROT$}}
  \rput(6,0){\rnode{B}{$\Tilde B$}}
  \rput(8,0){\rnode{Z}{$\calZ$}}
  \ncline{->}{Adv}{AC}
  \ncline[offset=2pt]{<-}{FROT}{B}\Aput[2pt]{$\scriptstyle c$}
  \ncline[offset=-3pt]{->}{FROT}{B}\Bput[2pt]{$\scriptstyle s_c$}
  \ncline{->}{AC}{FROT}\Aput[2pt]{$\scriptstyle s_0,s_1$}
  \ncline[offset=2pt]{<-}{B}{Z}\Aput[2pt]{$\scriptstyle c$}
  \ncline[offset=-3pt]{->}{B}{Z}\Bput[2pt]{$\scriptstyle s_c$}
  \ncloop[angleB=180,linearc=.2]{<->}{Z}{Adv}
}  

\newcommand\figureStandalone{
  \begin{pspicture}(-4.3,-6.7)(8.2,.7)
    \psset{nodesep=2pt}
    \rput(-4,0){\SAReal}
    \rput(0,0){\SARealX}
    \rput(5,0){\SARealXX}
    \rput(5,-4){\SAIdealXX}
    \rput(0,-4){\SAIdealX}
    \rput(-4,-4){\SAIdeal}
  \end{pspicture}}

\newcommand\SAReal{%
  \rput(0,0){\rnode{Z'}{$\calZ'$}}
  \rput(1,0){\rnode{D}{$D$}}
  \rput(0,-1){\rnode{Adv}{$\Adv$}}
  \rput(1,-1){\rnode{pi}{$A$}}
  \pnode(2,0){out}
  \psframe[linestyle=dashed,framearc=.2](-.3,.3)(1.3,-.3)
  \ncline{->}{Z'}{D}
  \ncline{->}{D}{out}
  \ncline{<->}{Z'}{Adv}
  \ncline{<->}{Adv}{pi}
  \ncline{->}{pi}{D}
  \rput(-.1,.49){$\calZ$}
  \rput(.5,-1.9){$(\mathit{Real})$}
}

\newcommand\SARealX{%
  \rput(0,0){\rnode{Z'}{$\calZ'$}}
  \rput(0,-1){\rnode{Adv}{$\Adv$}}
  \rput(1,-1){\rnode{pi}{$A$}}
  \pnode(1.7,-.5){out}
  \ncline{<->}{Z'}{Adv}
  \ncline{<->}{Adv}{pi}
  \ncline[nodesepB=5pt]{->}{Z'}{out}
  \ncline[nodesepB=5pt]{->}{pi}{out}
  \rput(.5,-1.9){$(\mathit{Real'})$}
}

\newcommand\SARealXX{%
  \rput(0,0){\rnode{Z'}{$\calZ'$}}
  \rput(0,-1){\rnode{Adv}{$\Adv$}}
  \rput(1,-1){\rnode{pi}{$A$}}
  \pnode(1.7,-.5){out}
  \ncline{<->}{Z'}{Adv}
  \ncline{<->}{Adv}{pi}
  \ncline[nodesepB=5pt]{->}{Z'}{out}
  \ncline[nodesepB=5pt]{->}{pi}{out}
  \rput(.2,-1.9){$(\mathit{Real''})$}
  \psframe[linestyle=dashed,framearc=.2](-.5,.3)(.5,-1.3)
  \rput(-.1,.56){$\calZ'_\Adv$}
}

\newcommand\SAIdealXX{%
  \rput(0,0){\rnode{Z'}{$\calZ'$}}
  \rput(0,-1){\rnode{Adv}{$\Adv$}}
  \rput(2.6,-1){\rnode{F}{$\calF$}}
  \rput(1.4,-1){\rnode{Sim*}{$\Sim^*$}}
  \pnode(3.3,-.5){out}
  \ncline{<->}{Z'}{Adv}
  \ncline{<->}{Adv}{Sim*}
  \ncline{<->}{Sim*}{F}
  \ncline[nodesepB=5pt]{->}{Z'}{out}
  \ncline[nodesepB=5pt]{->}{F}{out}
  \rput(.7,-2.5){$(\mathit{Ideal''})$}
  \psframe[linestyle=dashed,framearc=.2](-.5,.3)(.5,-1.3)
  \rput(-.1,.56){$\calZ'_\Adv$}
  \pspolygon[linestyle=dashed,linearc=.2]
  (-.7,1)
  (-.7,-1.5)
  (.7,-1.5)
  (2,-1.5)
  (2,-.7)
  (.6,-.7)
  (.7,1)
  \rput[t](-.3,-1.6){$\Sim'$}
}

\newcommand\SAIdealX{%
  \rput(0,0){\rnode{Z'}{$\calZ'$}}
  \rput(0,-1){\rnode{Adv}{$\Adv$}}
  \rput(2.6,-1){\rnode{F}{$\calF$}}
  \rput(1.4,-1){\rnode{Sim*}{$\Sim^*$}}
  \pnode(3.3,-.5){out}
  \ncline{<->}{Z'}{Adv}
  \ncline{<->}{Adv}{Sim*}
  \ncline{<->}{Sim*}{F}
  \ncline[nodesepB=5pt]{->}{Z'}{out}
  \ncline[nodesepB=5pt]{->}{F}{out}
  \psframe[linestyle=dashed,framearc=.2](-.5,-.7)(2,-1.3)
  \rput(.7,-2.5){$(\mathit{Ideal'})$}
  \rput[t](-.2,-1.4){$\Sim$}
}

\newcommand\SAIdeal{%
  \rput(0,0){\rnode{Z'}{$\calZ'$}}
  \rput(1,0){\rnode{D}{$D$}}
  \rput(0,-1){\rnode{Adv}{$\Sim$}}
  \rput(1,-1){\rnode{pi}{$\calF$}}
  \pnode(2,0){out}
  \psframe[linestyle=dashed,framearc=.2](-.3,.3)(1.3,-.3)
  \ncline{->}{Z'}{D}
  \ncline{->}{D}{out}
  \ncline{<->}{Z'}{Adv}
  \ncline{<->}{Adv}{pi}
  \ncline{->}{pi}{D}
  \rput(-.1,.49){$\calZ$}
  \rput(.5,-2.5){$(\mathit{Ideal})$}
}

\newcommand\DAReal{%
  \rput(0,0){\rnode{Z1}{$\calZ$}}
  \rput(1,0){\rnode{Adv1}{$\Adv$}}
  \rput(0,-1){\rnode{pi1}{$\pi$}}
  \ncline{-}{Z1}{Adv1}
  \ncarc[arcangle=45]{-}{Adv1}{pi1}
  \ncline{-}{pi1}{Z1}
  \rput(.5,-1.9){(\rmOne)}
}

\newcommand\DArouted{%
  \rput(0,0){\rnode{Z2}{$\calZ$}}
  \rput(1,0){\rnode{Adv2}{$\Adv$}}
  \rput(0,-1){\rnode{pi2}{$\pi$}}
  \rput(1,-1){\rnode{DA2}{\rlap{$\Advd$}\phantom{$\Adv$}}}
  \ncline{-}{Z2}{Adv2}
  \ncline{-}{pi2}{Z2}
  \ncline{-}{Adv2}{DA2}
  \ncline{-}{DA2}{pi2}
  \psframe[linestyle=dashed,framearc=.2](-.3,.3)(1.5,-.3)
  \rput(.1,.5){$\calZ_\Adv$}
  \rput(.8,-1.9){(\rmTwo)}
}

\newcommand\DAidealA{%
  \rput(0,0){\rnode{Z3}{$\calZ$}}
  \rput(1,0){\rnode{Adv3}{$\Adv$}}
  \rput(0,-1){\rnode{rho3}{$\rho$}}
  \rput(1,-1){\rnode{SimP3}{$\Sim'$}}
  \ncline{-}{Z3}{Adv3}
  \ncline{-}{rho3}{Z3}
  \ncline{-}{Adv3}{SimP3}
  \ncline{-}{SimP3}{rho3}
  \psframe[linestyle=dashed,framearc=.2](-.3,.3)(1.5,-.3)
  \rput(.1,.5){$\calZ_\Adv$}
  \rput(.5,-1.9){(\rmThree)}
}

\newcommand\DAidealB{%
  \rput(0,0){\rnode{Z4}{$\calZ$}}
  \rput(1,0){\rnode{Adv4}{$\Adv$}}
  \rput(0,-1){\rnode{rho4}{$\rho$}}
  \rput(1,-1){\rnode{SimP4}{$\Sim'$}}
  \ncline{-}{Z4}{Adv4}
  \ncline{-}{rho4}{Z4}
  \ncline{-}{Adv4}{SimP4}
  \ncline{-}{SimP4}{rho4}
  \psframe[linestyle=dashed,framearc=.2](.5,.3)(1.5,-1.4)
  \rput(.8,.49){$\Sim$}
  \rput(.7,-1.9){(\rmFour)}
}

\newcommand\figureDummyAdversary{%
  \begin{pspicture}(-4.5,-2)(7,0.6)
    \psset{nodesep=2pt}
    \rput(-4,0){\DAReal}
    \rput(-1,0){\DArouted}
    \rput(2.5,0){\DAidealA}
    \rput(5.3,0){\DAidealB}
  \end{pspicture}}

\newcommand\figureComposition{%
  \begin{pspicture}(-.5\figurewidth,-1.2)(.5\figurewidth,3)
    \def\twomach##1##2##3{\rput(0,0){\rnode{##1l}{##2}}\rput(1.4,0){\rnode{##1r}{##3}}\ncline{-}{##1l}{##1r}}
    \psset{nodesep=2pt}
    \rput(-5.7,2){(\rmOne)}\rput(-4,1){\realNet}
    \rput(-5.7,0.15){(\rmTwo)}\rput(-4,-1){\idealNet}
    \rput(.7,1){(\rmThree)}\rput(2.4,0){\hybNet}
  \end{pspicture}}

\newcommand\realNet{%
  \rput(-.5,1.7){\rnode{Z}{$\calZ$}}
  \rput(-.5,0.7){\rnode{A}{$\Adv$}}
  \rput(0.4,1.5){\rnode{sigma}{$\sigma$}}
  \rput(1.4,1.7){\rnode{pi1}{$\pi_1$}}
  \rput(1.4,1.3){$\vdots$}
  \rput(1.4,0.7){\rnode{pin}{$\pi_n$}}
  \ncline{-}{Z}{sigma}
  \ncline{-}{Z}{A}
  \ncline{-}{A}{sigma}
  \ncline{-}{sigma}{pi1}
  \ncline{-}{sigma}{pin}
  \ncline{-}{A}{pin}
  \ncarc[arcangle=-20]{-}{A}{pi1}
}

\newcommand\idealNet{%
  \rput(-.5,1.7){\rnode{Z}{$\calZ$}}
  \rput(-.5,0.7){\rnode{A}{$\Adv$}}
  \rput(0.4,1.2){\rnode{sigma}{$\sigma$}}
  \rput(1.4,1.7){\twomach{rho1}{$\Sim'_1$}{$\rho_1$}}
  \rput(1.4,1.3){$\vdots$}
  \rput(2.8,1.3){$\vdots$}
  \rput(1.4,0.7){\twomach{rhon}{$\Sim'_n$}{$\rho_n$}}
  \ncline{-}{Z}{sigma}
  \ncline{-}{Z}{A}
  \ncline{-}{A}{sigma}
  \ncline{-}{sigma}{rho1r}
  \ncline{-}{sigma}{rhonr}
  \ncline{-}{A}{rhonl}
  \ncarc[arcangle=-25]{-}{A}{rho1l}
  \pspolygon[linestyle=dashed,linearc=.2]
  (-1,1)
  (-1,.3)
  (1.9,.3)
  (1.9,2)
  (.85,2)
  (.85,1)
  \rput(0,-.1){\rput[t](-.68,.3){$\Sim$}}
}

\newcommand\hybNet{%
    \rput(-.5,1.7){\rnode{Z}{$\calZ$}}
    \rput(-.5,0.7){\rnode{A}{$\Adv$}}
    \rput(0.4,2){\rnode{sigma}{$\sigma$}}
    \rput(3.4,2.57){\rnode{pi1}{$\pi_1$}}
    \rput(3.4,2.2){$\vdots$}
    \rput(3.4,1.6){\rnode{im1}{$\pi_{i-1}$}}
    \rput(2.8,1){\twomach{i}{$\Advd$}{$\pi_i$}}
    \rput(2,0.3){\twomach{i1}{$\Sim'_{i+1}$}{$\rho_{i+1}$}}
    \rput(2,-.1){$\vdots$}
    \rput(3.4,-.1){$\vdots$}
    \rput(2,-.7){\twomach{n}{$\Sim'_n$}{$\rho_{n}$}}
    \ncline{-}{sigma}{pi1}
    \ncline{-}{sigma}{ir}
    \ncarc[arcangle=-37]{-}{sigma}{nr}
    \ncline{-}{Z}{sigma}
    \ncline{-}{Z}{A}
    \ncline{-}{A}{sigma}
    \ncline{-}{A}{pi1}
    \ncline{-}{A}{il}
    \ncarc[arcangle=-20]{-}{A}{nl}
    \pspolygon[linestyle=dashed,linearc=.2]
    (-1,2.9)
    (-1,-1.1)
    (3.9,-1.1)
    (3.9,0.65)(1.65,0.65)(1.65,1.3)(3.9,1.3)
    (3.9,2.9)
    \rput(.1,-.1){\rput[tl](-1,2.9){$\calZ_{\sigma,i}$}}
}

\newdimen\figurewidth
\begin{document}

\title{~\\[-1cm]Universally Composable \\Quantum Multi-Party Computation%
  \fullonly{\thanks{Funded by the Cluster of Excellence
    ``Multimodal Computing and Interaction''.}}}

\fullonly{\author{Dominique Unruh\\Saarland University}}

\date{}

\maketitle

\shortonly{\vspace*{-1.3cm}}

\shortonly{
  \renewenvironment{abstract}{\begin{quotation}\noindent\textbf{Abstract.}}{\end{quotation}}
}

\begin{abstract}
  \noindent The Universal Composability model (UC) by Canetti (FOCS 2001) 
  allows for secure composition of arbitrary protocols. We present a
  quantum version of the UC model which enjoys the same
  compositionality guarantees.  We
  prove that  in this model statistically secure oblivious transfer protocols can be
  constructed from commitments. Furthermore, we show
  that every statistically classically UC secure protocol is also
  statistically quantum UC secure. Such implications are not known for
  other quantum security definitions. As a corollary, we get that
  quantum UC secure protocols for general multi-party computation can
  be constructed from commitments.
  \shortonly{\par\noindent\textbf{Keywords:} Quantum cryptography,
      Universal Composability, oblivious transfer}
\end{abstract}

\pagestyle{plain}

\begin{fullversion}
{\columnseprule.4pt
\columnsep30pt
\setcounter{tocdepth}{2}
\hfill
\tableofcontents}
\newpage
\end{fullversion}

\section{Introduction}

Since the inception of quantum key distribution by Bennett and
Brassard \cite{Bennett:1984:Quantum}, it has been known that 
quantum communication permits to achieve protocol tasks that are
impossible given only a classical channel. For example, a
quantum key distribution\index{quantum key distribution}\index{key
  distribution!quantum}\index{QKD|see{quantum key distribution}} scheme \cite{Bennett:1984:Quantum} permits
to agree on a
secret key that is statistically secret, using only an authenticated but not secret channel. (By statistical security we
mean security against computationally unbounded adversaries, also
known as information-theoretical\index{information-theoretical|see{statistical}} security.) In contrast,
when using only classical communication, it is easy to see that such a
secret key can always be extracted by a computationally sufficiently
powerful adversary. Similarly, based on an idea by Wiesner
\cite{wiesner83conjugate}, Bennett, Brassard, Crépeau, and
Skubiszewska \cite{BBCS91} presented a protocol that was supposed to
construct an statistically secure oblivious
transfer\footnote{In an oblivious transfer protocol, Alice holds two
  bitstrings $m_0,m_1$, and Bob a bit $c$. Bob is supposed to get
  $m_c$ but not $m_{1-c}$, and Alice should not learn $c$.}  protocol
from a commitment, another feat that is easily seen to be impossible
classically.\footnote{We remark that,
on the other hand, Mayers \cite{Mayers:1996:Commitment} shows that
also in the quantum case, constructing an statistically secure 
commitment scheme \emph{without any additional assumption} is
impossible.
However, under additional assumptions like in the quantum bounded storage
model by Damgård, Fehr, Salvail, and Schaffner
\cite{Damgaard:2005:BoundedQuantum}, statistically secure bit commitment is possible.
See \autoref{sec:howto} for a discussion of the implications of
Mayers' impossibility result for
our result.}
Oblivious transfer, on the other hand, has been
recognized by Kilian \cite{Kilian:1988:Founding} to securely evaluate
arbitrary functions.  Unfortunately, the protocol of Bennett et
al.~could, at the time, not be proven secure, and the first complete
proof of (a variant of) that protocol was given almost two decades
later by Damgård, Fehr, Lunemann, Salvail, and Schaffner
\cite{Damgaard:2009:Improving}.  

Yet, although the oblivious transfer protocol
satisfies the intuitive secrecy requirements of oblivious transfer, in
certain cases the protocol might lose its security when used in a
larger context. In other words, there are limitations on how the
protocol can be composed. For example, no security guarantee is given
when several instances of the protocol are executed concurrently (see
\fullshort{\autoref{sec:restrict}}{Appendix~\ref{sec:restrict}} for a more detailed explanations of the various restrictions).

The problem of composability has been intensively studied by the
classical cryptography community (here and in the following, we use
the word classical as opposed to quantum). To deal with this problem
in a general way, Canetti \cite{Canetti:2001:Security} introduced the
notion of Universal Composability, short UC (Pfitzmann and Waidner
\cite{PfWa01} independently introduced the equivalent Reactive
Simulatability framework). The UC framework allows to express the
security of a multitude of protocol tasks in a unified way, and any UC-secure protocol automatically enjoys strong composability
guarantees (so-called universal composability). In particular, such a protocol can be run concurrently
with others, and it can be used as a subprotocol of other protocols in
a general way. Ben-Or and Mayers \cite{Ben-Or:2004:General} and
Unruh \cite{Unruh:2004:Simulatable} have shown that the idea of UC-security can be easily adapted to the quantum setting and have
independently presented quantum variants of the UC notion. These
notions enjoy the same strong compositionality guarantees. Shortly afterwards,
Ben-Or, Horodecki, Leung, Mayers, and Oppenheim
\cite{Ben-Or:2004:Universal} showed that many quantum key distribution
protocols are quantum-UC-secure. 

\paragraph{Our contribution.}
In this work, we use the UC framework to show the existence of a
statistically secure and universally composable oblivious
transfer protocol that uses only a commitment scheme. Towards this
goal, we first present a new definition of quantum-UC-security. In our
opinion, our notion is technically simpler than the notions of Ben-Or
and Mayers \cite{Ben-Or:2004:General} and Unruh
\cite{Unruh:2004:Simulatable}. We believe that this may also help to
increase the popularity of this notion in the quantum cryptography
community and to show the potential for using UC-security in the
design of quantum protocols. Second, we show that a variant of the
protocol by Bennett et at.~\cite{BBCS91} is indeed a UC-secure
oblivious transfer protocol.  By composing this protocol with a UC-secure protocol for general multi-party computations by Ishai,
Prabhakaran, and Sahai \cite{Ishai:2008:OT}, we get UC-secure
protocols for general multi-party computations using only commitments
and a quantum channel~-- this is easily seen to be
impossible in a purely classical setting.

\begin{fullversion}
  \fullshort\subsection\dotparagraph{Quantum Universal Composability
    (quantum-UC)} \fullonly{\label{sec:into.quc}} We begin by giving
  an overview over the UC framework. The basic idea behind the UC
  framework is to define security by comparison. Given a certain
  protocol task, say to implement a secure message
  transfer\index{secure message transfer}\index{message
    transfer!secure}, we first specify a machine, called the ideal
  functionality $\calF$ that, by definition, fulfills this protocol
  task securely. E.g., In the case of a secure message transfer, this
  functionality would take a value $x$ from Alice and give this value
  to Bob. All communication between parties and the functionality is
  done over secure channels. Obviously, this functionality $\calF$
  does exactly what we expect from a secure message transfer. Then, we
  define what it means for a protocol $\pi$ to be a secure
  implementation of $\calF$. Intuitively, we require that $\pi$ is no
  less secure than $\calF$. In other words, anything the adversary can
  do in an execution of $\pi$, the adversary could also do in an
  execution using $\calF$. (And in particular, since $\calF$ is secure
  by definition, the adversary then cannot perform any successful
  attacks on $\pi$ either.) This requirement is formally captured by
  requiring that for any adversary $\Adv$, there is another adversary
  $\Sim$, the simulator, such that an execution of $\pi$ with $\Adv$
  (called the real model) is indistinguishable from an execution of
  $\calF$ with $\Sim$ (called the ideal model). And
  indistinguishability in turn is modeled by requiring that no machine
  $\calZ$, called the environment, can guess whether it is interacting
  with the real model or with the ideal model. More precisely:
  \begin{definition}[(Classical) Universal Composability -- informal]
    We say $\pi$ classical-UC-emulates $\calF$ if for any adversary
    $\Adv$ there is a simulator $\Sim$ such that for all environments
    $\calZ$ we have that the difference between the following
    probabilities is negligible: The probability that $\calZ$ outputs
    $1$ in an execution of $\calZ$, $\Adv$, and $\pi$, and the
    probability that $\calZ$ outputs $1$ in an execution of $\calZ$,
    $\Sim$, and $\calF$. (We assume that $\calZ$ can freely
    communicate with the adversary/simulator.)
  \end{definition}
  In the example of a secure message transfer functionality $\calF$,
  the functionality would require its inputs $x$ from $\calZ$ and then
  send $x$ back to $\calZ$. In a secure message transfer protocol
  $\pi$, that is, in a protocol $\pi$ classical-UC-emulating $\calF$,
  Alice would than have to take the input $x$ from $\calZ$, and Bob
  would have to output $x$ to $\calZ$ (otherwise $\calZ$ could
  trivially distinguish the real and the ideal model).  All
  communication send between Alice and Bob over insecure channels,
  however, would be under the control of the adversary. Thus
  everything the adversary learns from that communication, the
  simulator would have to be able to produce on its own; in
  particular, the adversary cannot derive $x$ from that communication
  since the simulator could not simulate that knowledge (in the ideal
  model, the simulator cannot get $x$). This captures the intuitive
  requirement that a secure message transfer protocol should not
  reveal the message to the adversary. In a similar way, other
  properties like the authenticity of the message can be derived from
  the UC definition.

  The UC definition comes in many flavors. For example, in
  computational classical UC-security we restrict the adversary,
  simulator, and environment to polynomial-time machines. This variant
  is used if we want to show security based on computational
  assumptions.  In statistical classical UC, on the other hand, we
  quantify over all (possibly unbounded) adversaries, simulators, and
  environments. This variant is used to model statistical security.

  Besides providing a unified way to model the security of various
  protocol tasks by specifying the ideal functionality, the UC
  framework allows for very general composition of protocols. Assume a
  protocol~$\sigma^\calF$ that uses a functionality $\calF$ as a
  building block.  That is, in the real model, $\sigma^\calF$ has
  access to a functionality $\calF$ that performs a certain task in a
  fully trusted way. (We say, $\sigma^\calF$ runs in the
  $\calF$-hybrid model.) Assume that $\sigma^\calF$
  classical-UC-emulates some other functionality~$\calG$ and that we
  are given a protocol $\pi$ that classical-UC-emulates~$\calF$. Then
  the so-called universal composition theorem states that
  $\sigma^\pi$, the protocol resulting from using the subprotocol
  $\pi$ instead of $\calF$, also classical-UC-emulates $\calG$. This
  even holds if $\sigma^\calF$ invokes many instances of $\calF$
  concurrently.  Such a composition theorem is very useful for proving
  the security of larger protocols in a modular way: One first
  abstracts away a subprotocol (here $\pi$) by replacing it by some
  functionality (here~$\calF$), leading to a simpler protocol
  $\sigma^\calF$ in the $\calF$-hybrid model that is more amenable to
  analysis. Then the protocol $\pi$ is analyzed separately.  It should
  be noted that it was shown by Lindell \cite{Lindell:2003:General}
  that no security notion weaker than (a particular variant of)
  classical UC can have such a composition theorem.

  To get a variant of the UC notion suitable for modeling quantum
  cryptography, we only need to slightly modify the definition:
  Instead of quantifying over classical adversaries, simulators, and
  environment, we quantify over quantum adversaries, simulators, and
  environment. That is, the protocol parties, the adversary, the
  simulator, and the environment are allowed to store, send, and
  compute with quantum states. (And in the computational variant of
  quantum-UC-security, we additionally restrict adversaries,
  simulators, and environment to be restricted to polynomial-time
  quantum computations.)  Since, in a sense, we only change the
  machine model, most structural theorems about UC-security, in
  particular the universal composition theorem, still hold for
  quantum-UC-security; their proofs are almost identical in the
  classical and in the quantum setting. We present our model of
  quantum-UC in \autoref{sec:quc} and give a universal composition
  theorem for that model.
\end{fullversion}

\fullshort\subsection\dotparagraph{UC-secure quantum oblivious transfer}
The oblivious transfer (OT)\index{OT} protocol used in this paper is essentially
the same a the protocol proposed by Damgård et~al.~\cite{Damgaard:2009:Improving} which in turn is based on a
protocol by Bennett et~al.~\cite{BBCS91}.  The basic idea of the
protocol is that Alice encodes a random sequence $\tilde x$ of bits as a quantum
state, each bit randomly either in the computational basis or in the
diagonal basis.\footnote{If we were to use photons for transmission,
  in the computational basis we might encode the bit~$0$ as a
  vertically polarized photon and the bits~$1$ as a horizontally
  polarized photon. In the diagonal basis we might encode the bit~$0$
  as a $45$°-polarized photon, and the bit~$1$ as a $135$°-polarized
  photon.}  Then Bob is supposed to measure all bits, this time in
random bases of his choosing. Then Alice sends the bases she used to
Bob. Let $I_=$ denote the indices of the bits $\tilde x_i$ where Alice and Bob
chose the same basis, and $I_{\neq}$ the indices of the bits where
Alice and Bob chose different bases.  Assume that Bob wants to receive
the message $m_c$ out of Alice's messages $m_0,m_1$. Then Bob sets
$I_c:=I_=$ and $I_{1-c}:=I_{\neq}$ and sends $(I_0,I_1)$ to
Alice. Alice will not know which of these two sets is which and hence
does not learn $c$. Bob will know the bits $\tilde x_i$ at indices
$i\in I_c$. But
even a dishonest Bob, assuming that he measured the whole quantum
state, will not know the bits at indices $I_{1-c}$ since he used the
wrong bases for these bits. Thus Alice uses the bits at $I_0$ to
mask her message $m_0$, and the bits at $I_1$ to mask her message
$m_1$. Then Bob can recover $m_c$ but not $m_{1-c}$. (To deal with the
fact that a malicious Bob might have partial knowledge about the bits
at $I_{1-c}$, we use so-called privacy amplification to extract a near
uniformly mask from these bits.)

The problem with this analysis is that we have assumed that a
malicious Bob measures the whole quantum state upon reception. But
instead, Bob could store the quantum state until he learns the bases
that Alice used, and then use these bases to measure all bits
$\tilde x_i$ accurately. Hence, we need to force a dishonest Bob to
measure all bits before Alice sends the bases. The idea of Bennett et
al.~\cite{BBCS91} is to introduce the following test: Bob has to
commit to the bases he used and to his measurement outcomes. Then
Alice picks a random subset of the bits, and Bob opens the commitments
on his bases and outcomes corresponding to this subset of bits. Alice
then checks whether Bob's measurement outcomes are consistent with
what Alice sent. If Bob does not measure enough bits, then he will
commit to the wrong values in many of the commitments, and
there will be a high probability that Alice detects this.

It was a long-standing open problem what kind of a commitment needs to
be used in order for this protocol to be secure. Damgård et
al.~\cite{Damgaard:2009:Improving} give criteria for the commitment
scheme under which the OT protocol can be proven to have so-called
stand-alone security; stand-alone security, however, does not give as
powerful compositionality guarantees as UC-security
(cf.~\autoref{sec:restrict} below).
In
order to achieve UC-security, we assume that the commitment is given
as an ideal functionality. Then we have to show UC-security in the
case of a corrupted Alice, and UC-security in the case of a corrupted
Bob. The case of a corrupted Alice is simple, as one can easily see
that no information flows from Bob to Alice (the commitment
functionality does, by definition, not leak any information about the
committed values). The case of a corrupted Bob is more complex and
requires a careful analysis about the amount of information that Bob
can retrieve about Alice's bits. Such an analysis has already been
performed by Damgård et al.~\cite{Damgaard:2009:Improving} in their
setting. Fortunately, we do not need to repeat the analysis. We show
that that under certain special conditions, stand-alone security
already implies UC-security. Since in the case of a
corrupted Bob, these conditions are fulfilled, we get the security in
the case of a corrupted Bob as a corollary from the work by Damgård et
al.~\cite{Damgaard:2009:Improving}.

In \autoref{sec:ot}, we show that the OT protocol by Damgård et
al.~\cite{Damgaard:2009:Improving}, when using an ideal functionality
for the commitment, is statistically quantum-UC-secure. Furthermore,
the universal composition theorem guarantees that we can replace the
commitment functionality by any quantum-UC-secure commitment protocol.

\fullshort\subsection\dotparagraph{Quantum lifting and multi-party computation}
We are now equipped with a statistically quantum-UC-secure OT protocol
$\piQOT$ in the commitment-hybrid model. As noted first by Kilian
\cite{Kilian:1988:Founding}, OT can be used for securely evaluating
arbitrary functions, short, OT is complete for multi-party
computation. Furthermore, Ishai, Prabhakaran, and
Sahai~\cite{Ishai:2008:OT} showed that for any functionality $\calG$
(even interactive functionalities that proceed in several rounds),
there is a classical protocol $\rho^{\FOT}$ in the OT-hybrid model
that statistically classical-UC-emulates~$\calG$. Thus, to get a
protocol for $\calG$ in the commitment-hybrid model, we simply replace
all invocations to $\FOT$ by invocations of the subprotocol $\piQOT$,
resulting in a protocol $\rho^{\piQOT}$. We then expect that the
security of $\rho^{\piQOT}$ follows directly using the universal
composition theorem (in its quantum variant). There is, however, one
difficulty: To show that $\rho^{\piQOT}$ statistically
quantum-UC-emulates $\calG$, the universal composition theorem
requires that the following premises are fulfilled: $\piQOT$
statistically quantum-UC-emulates $\FOT$, and $\rho^{\FOT}$
statistically quantum-UC-emulates~$\calG$. But from the result of
Ishai et al.~\cite{Ishai:2008:OT} we only have that $\rho^{\FOT}$
statistically \emph{classical}-UC-emulates~$\calG$. Hence, we first
have to show that the same result also holds with respect to quantum-UC-security. Fortunately, we do not have to revisit the proof of Ishai
et al., because we show the following general fact:
\begin{theorem}[Quantum lifting theorem -- informal]\label{theo:qlift.inform}
  If the protocols~$\pi$ and~$\rho$ are classical protocols,
  and~$\pi$ statistically classical-UC-emulates~$\rho$, then 
  and~$\pi$ statistically quantum-UC-emulates~$\rho$.
\end{theorem}
Combining this theorem with the universal composition theorem, we
immediately get that $\rho^{\piQOT}$ statistically quantum-UC-emulates
$\calG$. In other words, any multi-party computation can be performed
securely using only a commitment and a quantum-channel.  In contrast, we show
that in the classical setting a commitment is not even
sufficient to compute the AND-function.

We stress that a property like the quantum lifting theorem should not
be taken for granted. For example, for the so-called stand-alone model
as considered by Fehr and Schaffner \cite{fehr09composing}, no
corresponding property is known.  
A special case of security in the stand-alone model is the
zero-knowledge property: The question whether protocols that are
statistical zero-knowledge with respect to classical adversaries are
also zero-knowledge with respect to quantum adversaries has been
answered positively by Watrous \cite{watrous06zk} for particular
protocols, but is still open in the general case.

\subsection{How to interpret our result}%
\label{sec:howto}
We show that we can perform arbitrary statistically UC-secure
multi-party computations, given a quantum channel and a
commitment. However, Mayers \cite{Mayers:1996:Commitment} has shown
that, even in the quantum setting, statistically secure commitment
schemes do not exist, not even with respect to security notions much
weaker than quantum-UC-security. In the light of this result, the
reader may wonder whether our result is not vacuous.  To illustrate
why our result is useful even in the light of Mayers' impossibility result, we present four
possible application scenarios.

\paragraph{Weaker computational assumptions.}
The first application of our result would be to combine our protocols with a
commitment scheme that is only \emph{computationally} quantum-UC-secure. Of course, the
resulting multi-party computation protocol would then not be
\emph{statistically} secure any more. However, since commitment
intuitively seems to be a simpler task than oblivious transfer,
constructing a computationally quantum-UC-secure commitment scheme
might be possible using simpler computational assumptions, and our
result then implies that the same computational assumptions can be
used for general multi-party computation.

\paragraph{Physical setup.}\index{physical setup}\index{setup!physical}
One might seek a direct physical implementation of a commitment, such
as a locked strongbox (or an equivalent but technologically more advanced
construct). With our result, such a physical implementation would be
sufficient for general multi-party computation. In contrast, in a
classical setting one would be forced to try to find physical
implementations of OT. It seems that a commitment might be a simpler
physical assumption than OT (or at least an incomparable one). So our
result reduces the necessary assumptions when implementing general
multi-party computation protocols based on physical assumptions. Also,
Kent \cite{kent99commitment} proposes to build commitments based on
the fact that the speed of light\index{speed of
  light}\index{light!speed of} is bounded. Although it is not clear
whether his schemes are UC-secure (and in particular, how to model his
physical assumptions in the UC framework), his ideas might lead to a
UC-secure commitment scheme that then, using our result, gives general
UC-secure multi-party computation based on the limitation of the speed
of light.

\paragraph{Theoretical separation.} Our result can also be seen from
the purely theoretical point of view.  It gives a separation between
the quantum and the classical setting by showing that in the quantum setting,
commitment is complete for general statistically secure multi-party
computation, while in the classical world it is not. Such
separations -- even without practical applications -- may increase our
understanding of the relationship between the classical and the
quantum setting and are therefore arguably interesting in their own
right.

\paragraph{Long-term
  security.}\index{security!long-term}\index{long-term security} Müller-Quade and Unruh
\cite{muellerquade07long} introduce the concept of long-term UC-security. In a nutshell, long-term UC-security is a strengthening of
computational UC-security that guarantees that a protocol stays
secure even if the adversary gets unlimited computational power after
the protocol execution. This captures the fact that, while we might
confidently judge today's technology, we cannot easily make
predictions about which computational problems will be hard in the
future. Müller-Quade and Unruh show that (classically) long-term UC-secure commitment protocols exist given certain practical
infrastructure assumptions, so-called signature cards\index{signature card}. It is, however,
likely that their results cannot be extended to achieve general multi-party
computation.  Our result, on the other hand, might allow to overcome this
limitation: Assume that we show that the commitment protocol of
Müller-Quade and Unruh is also secure in a quantum variant of
long-term UC-security. Then we could compose that commitment protocol
with the protocols presented here, leading to long-term UC-secure
general multi-party protocols from signature cards.

\delaytext{compositionality restrictions}{
\fullshort\subsection\section{Compositionality restrictions in prior work}
\label{sec:restrict}

\fullshort{Above}{In the introduction}, we claimed that the results of prior work on commitment schemes
in the quantum setting have limitations concerning their composability
guarantees. We will now briefly explain in which cases composition is
possible in prior models, and what the restrictions are. All prior
results giving some kind of composability guarantee work in the some
variant of the so-called stand-alone model\index{security!stand-alone|see{stand-alone model}}. The basic idea of the
stand-alone model is similar to that of the UC model: We specify a
protocol $\pi$ and an ideal functionality $\calF$, and we say that
$\pi$ implements $\calF$ in the stand-alone model if for every
adversary $\Adv$ attacking $\pi$ (real model), there is a simulator
$\Sim$ attacking $\calF$ (ideal model), so that the real and the ideal
model are indistinguishable. But in contrast to the UC model,
indistinguishability of the real and the ideal model is not defined
with respect to an environment that tries to guess which model it is
interacting with. Instead, given fixed inputs for all honest parties,
we require that the output of the honest parties and of the adversary
(considered as a joint quantum state) is indistinguishable from the
output of the functionality and of the simulator (considered as a
joint quantum state). The notion of indistinguishability of quantum
states is then defined depending of the flavor of the stand-alone
model. Security in the stand-alone model is strictly weaker than
security in the UC model: the UC environment may introduce additional
dependencies between the messages send in the protocol and the
protocol inputs/outputs. For example, the environment could give a
message that has been sent over an insecure channel by Bob as initial
protocol input to Alice. Such dependencies are explicitly excluded in
the stand-alone model.

In the classical case, it has been shown by Canetti
\cite{Canetti:2000:Security} that the stand-alone model allows for
sequential composition. Sequential composition\index{sequential composition}\index{composition!sequential} means that we are
allowed to run several protocols or several instances of one protocol
one after the other without loosing security, but we are not allowed
to run them concurrently or interleave the protocol steps (as can
easily happen if the protocol parties are not careful about their
synchronization). Similar results have been obtained in the quantum
case by Wehner and Wullschleger \cite{Wehner:2008:Composable} and by Fehr
and Schaffner \cite{fehr09composing} for different variants of the
quantum stand-alone model. 

There are two main flavors of the quantum stand-alone model:
Statistical and computational security. In the first case, adversary
and simulator are allowed to be unlimited, and in the second case,
adversary and simulator are computationally bounded. Note that when
defined like this, statistical stand-alone security does not imply
computational stand-alone security because statistical stand-alone
security does not guarantee that the simulator corresponding to a
computationally bounded adversary is also computationally bounded. The
effect of this is slightly paradoxical: one can compose statistically
secure protocols with each other, and one can compose computationally
secure protocols with each other, but no guarantees are given if one
composes a computationally secure protocol with a statistically secure
protocol. 

We note that the problems arising from an unlimited simulator can be
avoided by simply strengthening the statistical stand-alone model and
requiring that the simulator is computationally bounded if the
adversary is. This is the approach we also take in our modeling of
statistical quantum-UC-security.

The protocols analyzed by Wehner and Wullschleger
\cite{Wehner:2008:Composable} and Fehr and Schaffner
\cite{fehr09composing} are proven secure in (different variants of)
the statistical stand-alone model. Furthermore, the simulator they
construct does not run in polynomial time, therefore their results do
not imply computational stand-alone security and the difficulties
outlined above apply. 

The situation concerning the OT protocol analyzed by Damgård, Fehr,
Lunemann, Salvail, and Schaffner \cite{Damgaard:2009:Improving} is
even more subtle. They prove that in the case of a corrupted recipient
Bob, their protocol is secure in the computational stand-alone
model. Furthermore, for a corrupted sender Alice, the protocol is
secure in the statistical stand-alone model with non-polynomial-time
simulator. Thus, the protocol can be composed sequentially with other
protocols that are computationally secure for corrupted Bob and
statistically secure for corrupted Alice; yet it cannot be composed
with protocols which are statistically secure for corrupted Bob and
computationally secure for corrupted Alice. In particular, the OT
protocol cannot be composed with another instance of itself where Bob
is the sender. The full version \cite[Section
5]{Damgaard:2009:Improving:v3} of their paper describes an extension
of the underlying commitment scheme which enables the construction of an
efficient simulator. With such an
extension, sequential composition of their OT protocol with
computationally secure protocols is possible. 

In all three papers, when composing classical and quantum protocols,
it is necessary that even the classical protocols are proven secure
with respect to a definition involving quantum adversaries. A result
like our quantum lifting theorem (\autoref{theo:qlift.inform}) is an
open problem in the stand-alone model.
}
\fullonly{\usedelayedtext{compositionality restrictions}}

\begin{fullversion}
  \subsection{Related work}

\paragraph{Security models.} 
General quantum security models based on the stand-alone model have
first been proposed by van de Graaf \cite{Graaf:1998:Towards}. His
model comes without a composition theorem. The notion has been refined
by Wehner and Wullschleger \cite{Wehner:2008:Composable} and by Fehr
and Schaffner \cite{fehr09composing} who also prove sequential
composition theorems. Quantum security models in the style of the UC
model have been proposed by Ben-Or and Mayers
\cite{Ben-Or:2004:Universal} and by Unruh
\cite{Unruh:2004:Simulatable}.  The original idea behind the UC
framework in the classical setting was independently discovered by
Canetti \cite{Canetti:2001:Security} and by Pfitzmann and Waidner
\cite{PfWa01} (the notion is called Reactive Simulatability in the
latter paper).

\paragraph{Quantum protocols.} The idea of using quantum communication
for cryptographic purposes seems to originate from Wiesner
\cite{wiesner83conjugate}. The idea gained widespread recognition with
the BB84 quantum key-exchange protocol by Bennett and Brassard
\cite{Bennett:1984:Quantum}. A statistically hiding and binding
commitment scheme was proposed by Brassard, Cr\'epeau, Jozsa, and
Langlois \cite{brassard93commitment}. Unfortunately, the scheme was
later found to be insecure; in fact, Mayers
\cite{Mayers:1996:Commitment} showed that statistically hiding and
binding quantum commitments are impossible without using additional
assumptions. Kent \cite{kent99commitment} circumvents this
impossibility result by proposing a statistically hiding and binding
commitment scheme that is based on the limitation of the speed of
light.  Bennett, Brassard, Cr\'epeau, and Skubiszewska \cite{BBCS91}
present a protocol for statistically secure oblivious transfer in the
quantum setting. They prove their protocol secure under the assumption
that the adversary cannot store qubits and measures each qubit
individually. They also sketch an extension that uses a commitment
scheme to make their OT protocol secure against adversaries that can
store and compute on quantum states. The protocol analyzed in the
present paper is, in its basic idea, that extension.  Yao
\cite{Yao:1995:Security} gave a partial proof of the extended OT
protocol. His proof, however, is incomplete and refers to a future
complete paper which, to the best of our knowledge, never appeared.
As far as we know, the first complete proof of a variant of that OT
protocol has been given by Damgård, Fehr, Lunemann, Salvail, and
Schaffner \cite{Damgaard:2009:Improving}; their protocol is secure in
the stand-alone model. Hofheinz and Müller-Quade
\cite{Hofheinz:2003:Paradox} conjectured that the extended OT protocol
by Bennett et al.~\cite{BBCS91} is indeed UC-secure; in the present
paper we prove this claim.  Damgård, Fehr, Salvail, and Schaffner
\cite{Damgaard:2005:BoundedQuantum} have presented OT and commitment
protocols which are statistically secure under the assumption that the
adversary has a bounded quantum storage capacity.

\paragraph{Classical vs.~quantum security.} To the best of our
knowledge, van de Graaf \cite{Graaf:1998:Towards} was the first to
notice that even statistically secure classical protocols are not
necessarily secure in a quantum setting. The reason is that the
powerful technique of rewinding the adversary is not available in the
quantum setting. Watrous \cite{watrous06zk} showed that in particular
cases, a technique similar to classical rewinding can be used. He uses
this technique to construct quantum zero-knowledge proofs. No general
technique relating classical and quantum security is known; to the
best of our knowledge, our quantum lifting theorem is the first such
result (although restricted to the statistical UC model).

\paragraph{Miscellaneous.} Kilian \cite{Kilian:1988:Founding} first
noted that OT is complete for general multi-party computation.  Ishai,
Prabhakaran, and Sahai \cite{Ishai:2008:OT} prove that this also holds
in the UC setting. Computationally secure UC commitment schemes have
been presented by Canetti and Fischlin
\cite{Canetti:2001:Commitments}.
\end{fullversion}

\subsection{Preliminaries}

\paragraph{General.}
A nonnegative function $\mu$ is called negligible\index{negligible} if for all $c>0$ and
all sufficiently large $k$, $\mu(k)<k^{-c}$. A nonnegative function $f$ is
called overwhelming\index{overwhelming} if $f\geq 1-\mu$ for some negligible
$\mu$. Keywords in typewriter font (e.g., \texttt{environment}) are
assumed to be fixed but arbitrary, distinct non-empty words in
$\bits*$. $\varepsilon\in\bits*$ denotes the empty word\index{empty word}\index{word!empty}. Given a
sequence $x=x_1,\dots, x_n$, and a set $I\subseteq\{1,\dots,n\}$,
$x_{|I}$ denote the sequence $x$ restricted to the indices $i\in T$.

\paragraph{Quantum systems.} We can only give a few terse overview
over the formalism used in quantum computing. For a thorough
introduction, we recommend the textbook by Nielsen and Chuang
\cite[Chap.~1--2]{NiCh_00}. A (pure) state\index{pure state}\index{state!pure} in a quantum system is
described by a vector $\ket\psi$ in some Hilbert space $\calH$. In
this work, we only use Hilbert spaces of the form $\calH=\setC^N$ for
some countable set $N$, usually $N=\bit$ for qubits or $N=\bits*$ for
bitstrings. We always assume a designated orthonormal basis
$\{\ket x:x\in N\}$ for each Hilbert space, called the computational
basis\index{computational basis}. The basis states $\ket x$ represent classical states (i.e.,
states without superposition). Given several separate subsystems
$\calH_1=\setC^{N_1},\dots,\calH_n=\setC^{N_n}$, we describe the joint
system by the tensor product
$\calH_1\otimes\dots\otimes\calH_n=\setC^{N_1\times\dots\times N_n}$.
We write $\bra\Psi$ for the linear transformation mapping $\ket\Phi$
to the scalar product $\braket\Psi\Phi$. Consequently, $\butter\Psi$
denotes the orthogonal projector on $\ket\Psi$.  We set
$\ket0_+:=\ket0$, $\ket1_+:=\ket1$,
$\ket0_\times:=\frac1{\sqrt2}(\ket0+\ket1)$, and
$\ket1_\times:=\frac1{\sqrt2}(\ket0-\ket1)$. For $x\in\bits n$ and
$\theta\in\{+,\times\}^n$, we define $\ket
x_\theta:=\ket{x_1}_{\theta_1}\otimes\dots\otimes\ket{x_n}_{\theta_n}$.

\paragraph{Mixed states.}\index{state!mixed}\index{mixed state} If a system is not in a single pure state, but
instead is in the pure state $\ket{\Psi_i}\in\calH$ with
probability $p_i$ (i.e., it is in a mixed state), we describe the system by a density
operator\index{density operator}\index{operator!density}
$\rho=\sum_i p_i\butter{\Psi_i}$ over $\calH$.  This representation
contains all physically observable information about the distribution
of states, but some distributions are not distinguishable by any
measurement and are represented by the same mixed state. The set of
all density operators is the set of all positive\footnote{We call an
  operator positive if it is Hermitean and has only nonnegative
  Eigenvalues.} operators $\calH$ with trace $1$, and is denoted
$\calP(\calH)$. Composed systems are descibed by operators in
$\calP(\calH_1\otimes\dots\otimes\calH_n)$.  In the following, when
speaking about (quantum) states, we always mean mixed states in the density
operator representation.  A mapping
$\calE:\calP(\calH_1)\to\calP(\calH_2)$ represents a physically
possible operation (realizable by a sequence of unitary
transformations, measurements, and initializations and removals of
qubits) iff it is a completely positive trace preserving
map.\footnote{A map $\calE$ is completely positive iff for all Hilbert
  spaces $\calH'$, and all positive operators $\rho$ over
  $\calH_1\otimes\calH'$, $(\calE\otimes\id)(\rho)$ is positive.} We
call such mappings superoperators\index{superoperator}\index{operator!super-}.  The superoperator $\Einit^m$ on
$\calP(\calH)$ with $\calH:=\setC^{\bits*}$ and $m\in\bits*$ is
defined by $\Einit^m(\rho):=\butter m$ for all $\rho$. 

\paragraph{Composed systems.}\index{composed systems} Given a superoperator $\calE$ on
$\calP(\calH_1)$, the superoperator $\calE\otimes\id$ operates on
$\calP(\calH_1\otimes\calH_2)$. Instead of saying ``we apply
$\calE\otimes\id$'', we say ``we apply $\calE$ to $\calH_1$''. If we say ``we
initialize $\calH$ with $m$'', we mean ``we apply $\Einit^m$ to
$\calH$''.  Given a state $\rho\in\calP(\calH_1\otimes\calH_2)$, let
$\rho_x:=(\butter x\otimes\id)\rho(\butter x\otimes\id)$. Then the outcome of
measuring $\calH_1$ in the computational basis is $x$ with probability
$\tr\rho_x$, and after measuring $x$, the quantum state is $\frac{\rho_x}{\tr\rho_x}$.
Since we will only performs measurements in the computational basis in
this work, we will omit the qualification ``in the computational basis''.
The terminology in this paragraph generalizes to systems composed of
more than two subsystems.

\paragraph{Classical states.}\index{classical state}\index{state!classical} Classical probability distributions
$P:N\to[0,1]$ over a countable set~$N$ are represented by density
operators $\rho\in\calP(\setC^N)$ with $\rho=\sum_{x\in N}P(x)\butter{x}$ where
$\{\ket x\}$ is the computational basis. We call a state classical if
it is of this form. We thus have a canonical isomorphism between the
classical states over $\setC^N$ and the probability distributions over~$N$. We call a superoperator
$\calE:\calP(\setC^{N_1})\to\calP(\setC^{N_2})$
classical\index{classical superoperator}\index{superoperator!classical} iff if there is a
randomized function $F:N_1\to N_2$ such that
$\calE(\rho)=\sum_{\substack{x\in N_1\\y\in N_2}}\Pr[F(x)=y]\cdot \bra
x\rho\ket x\cdot\butter y$.
Classical superoperators describe what can be realized with classical
computations.  An example of a classical superoperator on
$\calP(\setC^N)$ is
$\Eclass:\rho\mapsto\sum_x\bra x\rho\ket x\cdot\butter x$.
Intuitively, $\Eclass$ measures $\rho$ in the computational basis and then
discards the outcome, thus removing all superpositions from $\rho$.

\section{Quantum Universal Composability}
\label{sec:quc}

We now present our quantum-UC-framework. 
\fullonly{For a motivation of the
model, we refer to \autoref{sec:into.quc}.}
\shortonly{The basic idea of our definition is the same as that
  underlying Canetti's UC-framework \cite{Canetti:2001:Security}.  The
  main change is that we allow all machines to perform quantum
  computations and to send quantum states as messages.  
  For a gentler introduction into the ideas and intuitions underlying
  the UC-framework, we refer to~\cite{Canetti:2001:Security}.
}

\paragraph{Machine model.}
A machine\index{machine} $M$ is described by an
identity\index{identity!of a machine} $\id_M$ in $\bits*$ and a
sequence of superoperators $\calE^{(k)}_M$ ($k\in\setN$) on
$\Hstate\otimes\Hclass\otimes\Hquant$ with
$\Hstate,\Hclass,\Hquant:=\setC^{\bits*}$ (the \emph{state transition
  operators}\index{state transition operator}). The index $k$ in
$\calE^{(k)}_M$ denotes
the security parameter. The Hilbert space
$\Hstate$ represents the state kept by the machine between
invocations, and $\Hclass$ and $\Hquant$ are used both for incoming
and outgoing messages. Any message consists of a classical part stored
in $\Hclass$ and a quantum part stored in $\Hquant$.  If a machine
$\id_\mathit{sender}$ wishes to send a
message with classical part $m$ and quantum part $\ket\Psi$ to a
machine $\id_\mathit{rcpt}$, the machine $\id_\mathit{sender}$
initializes $\Hclass$ with $(\id_\mathit{sender},\id_\mathit{rcpt},m)$
and $\Hquant$ with $\ket\Psi$.  (See the definition of the network
execution below for details.)  The separation of messages into a
classical and a quantum part is for clarity only, all information
could also be encoded directly in a single register. If a machine does not
wish to send a message, it initializes $\Hclass$ and $\Hquant$ with $\varepsilon$.

A network\index{network} $\bfN$ is a set of machines with pairwise
distinct identities containing a machine~$\calZ$ with
$\id_\calZ=\mathtt{environment}$. We write $\ids_\bfN$ for the set of
the identities of the machines in~$\bfN$.

We call a machine $M$
quantum-polynomial-time\index{quantum-polynomial-time}\index{polynomial-time!quantum-}
if there is a uniform\footnote{A sequence of circuits $C_k$ is uniform
  if a deterministic Turing machine can output the description of~$C_k$ in time polynomial in $k$.}  sequence of quantum circuits
$C_k$ such that for all $k$, the circuit $C_k$ implements the
superoperator~$\calE^{(k)}_M$.

\paragraph{Network execution.}
The state space $\calH_\bfN$ for a network $N$ is defined as
$\calH_\bfN:=
\Hclass\otimes\Hquant\otimes\bigotimes_{\id\in\ids_\bfN}\Hstate_\id$
with $\Hstate_\id,\Hclass,\Hquant:=\setC^{\bits*}$. Here $\Hstate_\id$
represents the local state of the machine with identity $\id$ and
$\Hclass$ and $\Hquant$ represent the state spaces used for
communication. ($\Hclass$ and $\Hquant$ are shared between all
machines. Since only one machine is active at a time, no conflicts
occur.)

A step in the execution of $\bfN$ is defined by a superoperator
$\calE:=\calE_\bfN^{(k)}$ operating on $\calH_\bfN$. This
superoperator performs the following steps: First, $\calE$ measures
$\Hclass$ in the computational basis, and parses the outcome as
$(\id_\mathit{sender},\id_\mathit{rcpt},m)$.  Let $M$ be the machine
in $\bfN$ with identity $\id_\mathit{rcpt}$.  Then $\calE$ applies
$\calE_M^{(k)}$ to $\Hstate_\id\otimes\Hclass\otimes\Hquant$.  Then
$\calE$ measures $\Hclass$ and parses the outcome as
$(\idsender',\idrcpt',m')$.  If the outcome could not be parsed, or if
$\idsender'\neq\idrcpt$, initialize $\Hclass$ with
${(\varepsilon,\mathtt{environment},\varepsilon)}$ and $\Hquant$
with $\varepsilon$. (This ensures
that the environment is activated if a machine sends no or an
ill-formed message.) 

The output of the network $\bfN$ on input $z$ and security parameter
$k$ is described by the following algorithm: Let
$\rho\in\calP(\calH_\bfN)$ be the state that is initialized to
${(\varepsilon,\environment,z)}$ in $\Hclass$, and to the empty
word $\varepsilon$ in all other registers. Then repeat the
following indefinitely: Apply $\calE_\bfN^{(k)}$ to $\rho$.  Measure
$\Hclass$. If the outcome is of the form
$(\environment,\varepsilon,\mathit{out})$, return $\mathit{out}$ and
terminate. Otherwise, continue the loop. The probability distribution
of the return value $\mathit{out}$ is denoted by
$\Exec_\bfN(k,z)$. 

\paragraph{Corruptions.}\index{corruption}
To model corruptions, we introduce \emph{corruption
  parties}\index{corruption party}\index{party!corruption}, special
machines that follow the instructions given by the adversary.  When
invoked, the corruption party\pagelabel{page:corruption.party}~$P^C_\id$ with identity $\id$
measures $\Hclass$ and parses the outcome as $(\idsender,\idrcpt,m)$.
If $\idsender=\adversary$, $\Hclass$ is initialized with
${m}$. (In this case, $m$ specifies both the message and the
sender/recipient. Thus the adversary can instruct a corruption party
to send to arbitrary recipients.)  Otherwise, $\Hclass$ is initialized with
$(\id,\adversary,(\idsender,\idrcpt,m))$. (The message
is forwarded to the adversary.) Note that, since $P^C_\id$ does not
touch the $\Hquant$, the quantum part of the message is forwarded.
\fullonly\par
Given a network $\bfN$, and a set of identities $C$, we write
$\bfN^C$ for the set resulting from replacing each machine
$M\in\bfN$ with identity $\id\in C$ by $P^C_\id$.

\paragraph{Security model.} A protocol\index{protocol} $\pi$ is a set of machines with
$\mathtt{environment},\mathtt{adversary}\notin\ids(\pi)$.  We assume a
set of identities $\parties_\pi\subseteq\ids(\pi)$\index{party} to be associated
with $\pi$. $\parties_\pi$ denotes which of the machines in the
protocol are actually protocol parties (as opposed to incorruptible
entities such as ideal functionalities).  

An environment is a machine with identity
$\mathtt{environment}$, an adversary or a simulator is a machine with
identity $\mathtt{adversary}$ (there is no formal distinction between
adversaries and simulators, the two terms refer to different
intended roles of a machine).

In the following we call two networks
indistinguishable\index{indistinguishability!of networks}
if there is a negligible function $\mu$ such that
for all $z\in\bits*$ and $k\in\setN$,
$ \abs{ \Pr[\Exec_N(k,z)=1] - \Pr[\Exec_M(k,z)=1] } \leq \mu(k)$.
We speak of perfect
indistinguishability\index{indistinguishability!perfect}\index{perfect
  indistinguishability} if $\mu=0$.

\begin{definition}[Statistical quantum-UC-security]\label{def:stat.quc}\index{UC!statistical quantum}\index{statistical quantum
    UC}\index{quantum-UC!statistical}
  Let protocols $\pi$ and $\rho$ be given. We say $\pi$
  \emph{statistically quantum-UC-emulates $\rho$}%
  \index{quantum-UC-emulate!statistically}%
  \index{UC-emulate!statistically quantum-}%
  \index{emulate!statistically quantum-UC-}%
  \index{statistically quantum-UC-emulate}
  iff for every set
  $C\subseteq \parties_\pi$ and for every adversary $\Adv$ there is a
  simulator $\Sim$ such that for every environment $\calZ$, the networks
  $\pi^C\cup\{\Adv,\calZ\}$ (called the \index{model!real}\index{real
    model}real model) and $\rho^C\cup\{\Sim,\calZ\}$ (called the \index{model!ideal}\index{ideal
    model}ideal model) are
  indistinguishable.  We furthermore require that if $\Adv$ is
  quantum-polynomial-time, so is $\Sim$.
\end{definition}

\begin{definition}[Computational quantum-UC-security]\index{UC!computational quantum}\index{computational quantum
    UC}\index{quantum-UC!computational}\index{Universal Composability|see{UC}}
  Let protocols $\pi$ and $\rho$ be given. We say $\pi$
  \emph{computationally quantum-UC-emulates $\rho$}%
  \index{quantum-UC-emulate!computationally}%
  \index{UC-emulate!computationally quantum-}%
  \index{emulate!computationally quantum-UC-}%
  \index{computationally quantum-UC-emulate}
  iff for every set $C\subseteq \parties_\pi$
  and for every quantum-polynomial-time adversary $\Adv$ there is a
  quantum-polynomial-time simulator $\Sim$ such that for every
  quantum-polynomial-time environment $\calZ$, the networks
  $\pi^C\cup\{\Adv,\calZ\}$ and $\rho^C\cup\{\Sim,\calZ\}$ are
  indistinguishable.
\end{definition}

\noindent Note that although $\Exec_{\pi^C\cup\{\Adv,\calZ\}}(k,z)$ may return
arbitrary bitstrings, we only compare whether the return value of
$\calZ$ is $1$ or not. This effectively restricts $\calZ$ to returning
a single bit. This can be done without loss of generality (see
\cite{Canetti:2001:Security} for a discussion this issue; their
arguments also apply to the quantum case) and simplifies the definition.

In our framework, any communication between two parties is perfectly
secure since the network model guarantees that they are delivered to the
right party and not leaked to the adversary. To model a protocol with
insecure channels\index{channel!insecure}\index{insecure channel}
instead, one would explicitly instruct the protocol parties to send
all messages through the adversary. Authenticated
channels\index{authenticated channel}\index{channel!authenticated} can
be realized by introducing an ideal functionality (see the next
section) that realizes an authenticated channel. For simplicity, we
only consider protocols with secure
channels\index{channel!secure}\index{secure channel} in this work.

\fullshort\subsection\dotparagraph{Ideal functionalities}
In most cases, the behavior of the ideal model is described by a
single machine $\calF$, the so-called ideal
functionality\index{functionality}\index{ideal functionality|see{functionality}}. We can
think of this functionality as a trusted third party that perfectly
implements the desired protocol behavior. For example, the
functionality $\FOT$ for oblivious transfer would take as input from
Alice two bitstrings $m_0,m_1$, and from Bob a bit $c$, and send to
Bob the bitstring $m_c$. Obviously, such a functionality constitutes a
secure oblivious transfer. We can thus define a protocol $\pi$ to be a
secure OT protocol if $\pi$ quantum-UC-emulates $\FOT$ where $\FOT$
denotes the protocol consisting only of one machine, the functionality
$\FOT$ itself. There is, however, one technical difficulty here. In
the real protocol $\pi$, the bitstring $m_c$ is sent to the
environment $\calZ$ by Bob, while in a the ideal model, $m_c$ is sent
by the functionality.  Since every message is tagged with the sender
of that message, $\calZ$ can distinguish between the real and the
ideal model merely by looking at the sender of $m_c$. To solve this
issue, we need to ensure that $\calF$ sends the message $m_c$ in the
name of Bob (and for analogous reasons, that $\calF$ receives messages
sent by $\calZ$ to Alice or Bob). To achieve this, we use so-called
dummy-parties \cite{Canetti:2001:Security} in the
ideal model. These are parties with the identities of Alice and Bob
that just forward messages between the functionality and the
environment. 
\begin{definition}[Dummy-party]\label{def:dummy.party}\index{dummy-party}\index{party!dummy-}
  Let a machine $P$ and a functionality $\calF$ be given.  The
  dummy-party $\Tilde P$ for $P$ and~$\calF$
  is a machine that has the same identity as
  $P$ and has the following state transition operator: Let $\id_\calF$
  be the identity of $\calF$. When activated, measure $\Hclass$. If
  the outcome of the measurement is of the form
  $(\mathtt{environment},\id_P,m)$, initialize $\Hclass$ with
  $(\id_P,\id_\calF,m)$. If the outcome is of the form
  $(\id_\calF,\id_P,m)$, initialize $\Hclass$ with
  $(\id_P,\mathtt{environment},m)$. In all cases, the quantum communication
  register is not modified (i.e., the message in that register is
  forwarded).
\end{definition}
Note the strong analogy to the corruption parties
(\autopageref{page:corruption.party}).

Thus, if we write $\pi$ quantum-UC-emulates $\calF$, we mean that
$\pi$ quantum-UC-emulates $\rho_\calF$ where $\rho_\calF$ consists of
the functionality $\calF$ and the dummy-parties corresponding to the
parties in $\pi$. More precisely:
\begin{definition}
  Let $\pi$ be a protocol and $\calF$ be a functionality.  We say that
  $\pi$ statistically/computationally quantum-UC-emulates $\calF$ if
  $\pi$ statistically/computationally quantum-UC-emulates $\rho_\calF$
  where $\rho_\calF:=\{\Tilde P:P\in\parties_\pi\}\cup\{\calF\}$.
\end{definition}

\noindent For more discussion of dummy-parties and functionalities, see \cite{Canetti:2001:Security}.

Using the concept of an ideal functionality, we can specify a range of
protocol tasks by simply defining the corresponding
functionality. Below, we give the definitions of various
functionalities. All these functionalities are classical, we therefore
do not explicitly describe when the registers $\Hclass$ and $\Hquant$
are measured/initialized but instead describe the functionality in
terms of the messages sent and received.

\begin{definition}[Commitment]\label{def:com}\index{commitment!functionality}\index{functionality!commitment}
  Let $A$ and $B$ be two parties.  The functionality
  $\FCOM^{B\rightarrow A,\ell}$ behaves as follows: Upon (the first)
  input $(\mathtt{commit},x)$ with $x\in\{0,1\}^{\ell(k)}$ from $B$, send $\mathtt{committed}$ to
  $A$. Upon input $\mathtt{open}$ from $B$ send $(\mathtt{open},x)$ to $A$. 
  All communication/input/output is classical.
  \fullonly\par
  We call $B$ the sender and $A$ the recipient.
\end{definition}

\begin{definition}[Oblivious transfer
  (OT)]\label{def:ot}\index{OT!functionality}\index{oblivious transfer|see{OT}}\index{functionality!OT}
  Let $A$ and $B$ be two parties.  The functionality
  $\FOT^{A\rightarrow B,\ell}$ behaves as follows: When receiving input
  $(s_0,s_1)$ from $A$ with $s_0,s_1\in\{0,1\}^{\ell(k)}$ and $c\in\bit$
  from $B$, send $s:=s_c$ to $B$.  All communication/input/output is
  classical.
  \fullonly\par
  We call $A$ the sender and $B$ the recipient.\footnote{We used $A$
    as the sender in the description of the OT functionality, and as
    the recipient in the description of the commitment
    functionality. We do so to simplify notation later; our
    protocol for OT from $A$ to $B$ will use a commitment from $B$ to
    $A$.}
\end{definition}

\begin{definition}[Randomized oblivious transfer
  (ROT)]\label{def:rot}\index{randomized
    OT!functionality}\index{functionality!randomized OT}
  Let $A$ and $B$ be two parties.  The functionality
  $\FROT^{A\rightarrow B,\ell}$ behaves as follows: If $A$ is
  uncorrupted, when receiving input $c\in\bit$ from $B$, choose
  $s_0,s_1\in\bits{\ell(k)}$ uniformly and send $(s_0,s_1)$ to $A$ and
  $s:=s_c$ to $B$.  If $A$ is corrupted, when receiving input
  $(s_0,s_1)$ from $A$ with $s_0,s_1\in\{0,1\}^{\ell(k)}$ and
  $c\in\bit$ from $B$, send $s:=s_c$ to $B$.  All
  communication/input/output is classical.
\end{definition}

\fullonly{\subsection{Elementary properties of UC-security}
\label{sec:prop-uc-secur}}

\delaytext{lemma:ref.trans}{
\begin{lemma}[Reflexivity, transitivity]\label{lemma:ref.trans}\index{reflexivity}\index{transitivity}
  Let $\pi$, $\rho$, and $\sigma$ be protocols. Then $\pi$
  quantum-UC-emulates $\pi$. If $\pi$ quantum-UC-emulates $\rho$
  and $\rho$ quantum-UC-emulates $\sigma$, then $\pi$
  quantum-UC-emulates $\sigma$.
  \fullonly\par
  This holds both for statistical and computational
  quantum-UC-security.
\end{lemma}

\begin{proof}
  We first consider the case of statistical quantum-UC-security.

  For any adversary $\Adv$ and any set $C$, with $\Sim:=\Adv$, we have that
  $\pi^C\cup\{\Adv,\calZ\}$ and $\pi^C\cup\{\Sim,\calZ\}$ are equal
  and hence perfectly indistinguishable for all $\calZ$. If $\Adv$ is
  quantum-polynomial-time, so is $\Sim=\Adv$. Thus $\pi$
  quantum-UC-emulates $\rho$. 

  Assume that $\pi$ quantum-UC-emulates $\rho$ and $\rho$
  quantum-UC-emulates $\sigma$. Fix an adversary $\Adv$ and a set $C$. Then there is
  a simulator $\Sim$ such that for all $\calZ$,
  $\pi^C\cup\{\Adv,\calZ\}$ and $\rho^C\cup\{\Sim,\calZ\}$ are
  indistinguishable. Furthermore, for the adversary $\Adv':=\Sim$,
  there is a simulator $\Sim'$ such that
  $\rho^C\cup\{\Sim,\calZ\}=\rho^C\cup\{\Adv',\calZ\}$ and
  $\sigma^C\cup\{\Sim',\calZ\}$ are indistinguishable for all $\calZ$.
  Since indistinguishability is transitive, 
  $\pi^C\cup\{\Adv,\calZ\}$ and $\sigma^C\cup\{\Sim',\calZ\}$ are
  indistinguishable for all $\calZ$.  Finally, if $\Adv$ is quantum-polynomial-time,
  so is $\Adv'=\Sim$, and thus also $\Sim'$.  Thus $\pi$ quantum-UC-emulates~$\sigma$.

  In the case of computational quantum-UC-security, the proof is
  identical, except that we quantify over quantum-polynomial-time
  $\Adv$ and $\calZ$.
  \qed
\end{proof}
}
\fullonly{\usedelayedtext{lemma:ref.trans}}

\paragraph{Dummy-adversary.}\index{dummy-adversary}\index{adversary!dummy}
In the definition of UC-security, we have three entities interacting
with the protocol: the adversary, the simulator, and the
environment. Both the adversary and the environment are
all-quantified, hence we would expect that they do, in some sense,
work together. This intuition is backed by the following fact which
was first noted by Canetti \cite{Canetti:2001:Security}: Without loss
of generality, we can assume an adversary that is completely controlled by
the environment. This so-called dummy-adversary only forwards messages
between the environment and the protocol. The actual attack is then
executed by the environment.

\begin{definition}[Dummy-adversary $\Advd$]\label{def:adv.dummy}
  When activated, the dummy-ad\-ver\-sary $\Advd$ measures $\Hclass$;
  call the outcome $m$. If $m$ is of the form
  $(\environment,\penalty0\adversary,\penalty0 m')$, initialize $\Hclass$ with
  $m'$. Otherwise initialize $\Hclass$ with
  $(\adversary,\penalty0\environment,\penalty0m)$.  In all cases, the quantum
  communication register is not modified (i.e., the message in that
  register is forwarded).
\end{definition}

\noindent Note the strong analogy to the dummy-parties
(\autoref{def:dummy.party}) and the corruption parties
(\autopageref{page:corruption.party}).

\begin{lemma}[Completeness of the
  dummy-adversary]\label{lemma:dummy.complete}\index{dummy-adversary!completeness
  of}
  Assume that $\pi$ quantum-UC-emulates~$\rho$ with respect to the
  dummy-adversary (i.e., instead of quantifying over all adversaries
  $\Adv$, we fix $\Adv:=\Advd$). Then $\pi$ quantum-UC-emulates~$\rho$.
  \fullonly\par
  This holds both for statistical and computational quantum-UC-security.
\end{lemma}

\delaytext{proof lemma:dummy.complete}{
\fullshort{\begin{proof}}{\begin{proof}[of \autoref{lemma:dummy.complete}]}
  We first consider the case of statistical quantum-UC-security.

  \begin{figure}[t]
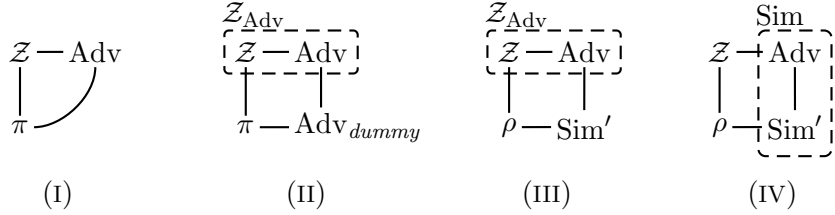

    \centering\figureDummyAdversary
    \caption{\label{fig:da}Completeness of the dummy-adversary: proof steps}
  \end{figure}

  Assume that $\pi$ statistically quantum-UC-emulates $\rho$ with
  respect to the dummy-adversary.  Fix an
  adversary $\Adv$. We have to show that there exists a
  simulator $\Sim$ such that for all
  environments $\calZ$ we have that
  $\pi\cup\{\Adv,\calZ\}$ and $\rho\cup\{\Sim,\calZ\}$ are
  indistinguishable.  Furthermore, if $\Adv$ is
  quantum-polynomial-time, $\Sim$ has to be quantum-polynomial-time,
  too.

  For a given environment $\calZ$, we construct an environment $\calZ_\Adv$ that is supposed to interact
  with $\Advd$ and internally simulates $\calZ$ and $\Adv$, and that
  routes all messages sent by the simulated $\Adv$ to $\pi$ through $\Advd$ and
  vice versa. Then $\pi\cup\{\Adv,\calZ\}$ and
  $\pi\cup\{\Advd,\calZ_\Adv\}$ are perfectly
  indistinguishable. (Cf.~networks (\rmOne) and (\rmTwo) in \autoref{fig:da}.) Since $\pi$ statistically
  quantum-UC-emulates $\rho$ with respect to the dummy-adversary, we
  have that $\pi\cup\{\Advd,\calZ_\Adv\}$ and
  $\rho\cup\{\Sim',\calZ_\Adv\}$ are indistinguishable for some
  $\Sim'$ and all
  $\calZ$. (Cf.~networks (\rmTwo) and~(\rmThree).) Since $\Advd$ is quantum-polynomial-time, so is
  $\Sim'$.
  We construct a machine $\Sim$ that internally simulates
  $\Sim'$ and $\Adv$ (network (\rmFour)). Then $\rho\cup\{\Sim',\calZ_\Adv\}$ and
  $\rho\cup\{\Sim,\calZ\}$ are perfectly indistinguishable. 
  Summarizing, 
  $\pi\cup\{\Adv,\calZ\}$ and $\rho\cup\{\Sim,\calZ\}$ are
  indistinguishable for all environments  $\calZ$. Furthermore, since $\Sim'$ is
  quantum-polynomial-time, we have that $\Sim$ is
  quantum-polynomial-time if $\Adv$ is. 
  This concludes the proof in the case of statistical quantum-UC-security.

  The proof in the case of computational quantum-UC-security is
  identical, except that we consider only quantum-polynomial-time $\Adv$ and
  $\calZ$, and thus have that $\calZ_\Adv$, $\Sim'$, and $\Sim$ are
  quantum-polynomial-time.  \qed
\end{proof}
}
\fullonly{\usedelayedtext{proof lemma:dummy.complete}}

\begin{shortversion}
  \noindent The proof of \autoref{lemma:dummy.complete} is very similar to that given in
  \cite{Canetti:2001:Security} and given in Appendix~\ref{app:proofs:sec:quc}.
\end{shortversion}

\fullshort\subsection\dotparagraph{Universal composition}
\fullonly{\label{sec:comp}}%
For some protocol $\sigma$, and some protocol~$\pi$, by $\sigma^\pi$
we denote the protocol where $\sigma$ invokes (up to polynomially
many) instances of $\pi$. That is, in $\sigma^\pi$ the machines from~$\sigma$ and from $\pi$ run together in one network, and the machines
from $\sigma$ access the inputs and outputs of $\pi$. (That is,
$\sigma$ plays the role of the environment from the point of view of~$\pi$. In particular,
$\calZ$ then talks only to $\sigma$ and not to the subprotocol $\pi$
directly.)   A typical
situation would be that $\sigma^\calF$ is some protocol that makes use
of some ideal functionality~$\calF$, say a commitment functionality, and then
$\sigma^\pi$ would be the protocol resulting from implementing that
functionality with some protocol $\pi$, say a commitment protocol. 
(We say that $\sigma^\calF$ is a protocol  in the $\calF$-hybrid
model\index{hybrid model}\index{model!hybrid}.)
One
would hope that such an implementation results in a secure protocol
$\sigma^\pi$. That is, we hope that if $\pi$ quantum-UC-emulates $\calF$ and~$\sigma^\calF$
quantum-UC-emulates $\calG$, then~$\sigma^\pi$ quantum-UC-emulates $\calG$. Fortunately,
this is the case:

  \begin{theorem}[Universal Composition
    Theorem]\label{theo:comp}\index{composition theorem}
    Let $\pi$, $\rho$, and $\sigma$ be quantum-polynomial-time
    protocols. Assume that $\pi$ quantum-UC-emulates $\rho$. Then
    $\sigma^\pi$ quantum-UC-emulates~$\sigma^\rho$.
    \fullonly\par
    This holds both for statistical and computational
    quantum-UC-security.
  \end{theorem}

\noindent If we additionally have that $\sigma$ quantum-UC-emulates $\calG$,
from the transitivity of quantum-UC-emulation
(\autoref{lemma:ref.trans}\shortonly{ in Appendix~\ref{app:proofs:sec:quc}}), it immediately follows that $\sigma^\pi$
quantum-UC-emulates~$\calG$.

  \begin{fullversion}
    The compositionality guarantee given by \autoref{theo:comp} is
    often called \emph{universal composability}. One should not
    confuse universal composability with UC-security. Although UC
    security implies universal composability, it has been shown by
    Hofheinz and Unruh \cite{Hofheinz:2005:Comparing,
      Hofheinz:2006:Concurrent, Unruh:2006:Protokollkomposition} that~-- in the classical setting at least --
    universal composability is a strictly weaker notion than UC
    security.
  \end{fullversion}

  \delaytext{proof theo:comp}{
  \begin{figure}[t]
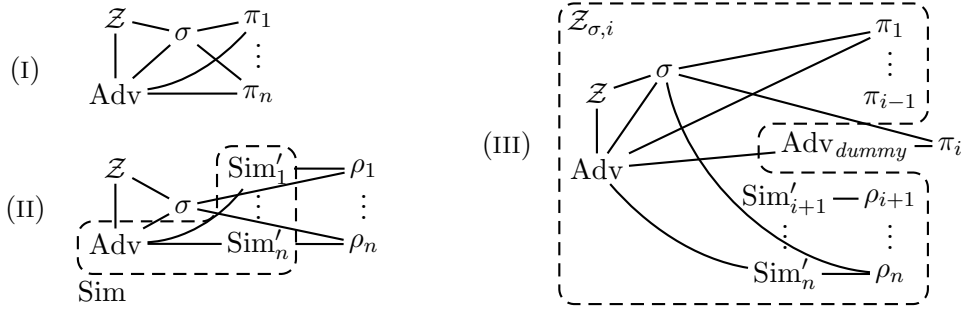

    \centering\figureComposition
    \caption{\label{fig:comp}Networks occurring in the proof sketch of
      \autoref{theo:comp}. Network (\rmOne) represents the real model,
      (\rmTwo) the ideal model, and (\rmThree) the hybrid case. To
      avoid cluttering, in (\rmThree), the connections to $\pi_{i-1}$,
      $\Sim'_{i+1}$, and $\rho_{i+1}$ have been omitted.}
  \end{figure}

  \begin{proof}[of \autoref{theo:comp}]
    We first show  \autoref{theo:comp} for the case of computational quantum-UC-security. Thus,
    our goal is to prove that under the assumptions of
    \autoref{theo:comp}, $\sigma^\pi$ computationally quantum-UC-emulates
    $\sigma^\rho$. 
    Since $\sigma$ is quantum-polynomial-time, $\sigma$ invokes at
    most a polynomial number $n$ of instances of its subprotocol $\pi$
    or $\rho$.
    Since $\pi$
    quantum-UC-emulates $\rho$, there is a quantum-polynomial-time simulator
    $\Sim'$ such that for all environments $\calZ$ we have that
    $\pi\cup\{\Advd,\calZ\}$ and $\rho\cup\{\Sim',\calZ\}$ are
    indistinguishable. In the following, we call $\Sim'$ the
    dummy-simulator.

  Let a quantum-polynomial-time adversary $\Adv$ be given (that is supposed
  to attack~$\sigma^\pi$). We construct a simulator $\Sim$ that
  internally simulates the adversary $\Adv$ and $n$ instances
  $\Sim'_1,\dots,\Sim'_n$ of the dummy-simulator
  $\Sim'$. The simulated adversary $\Adv$ is connected to the
  environment and to the protocol $\sigma$, but all messages between
  $\Adv$ and the $i$-th instance $\pi_i$ of $\pi$ are routed through
  the dummy-simulator-instance $\Sim'_i$ (which is then supposed
  to transform these messages into a form suitable for instances of
  $\rho$). The simulator $\Sim$ is depicted by the dashed box in network
  (\rmTwo) in \autoref{fig:comp}.
  
  We have to show that for any environment $\calZ$ we have that
  $\sigma^\pi\cup\{\Adv,\calZ\}$ and
  $\sigma^\rho\cup\{\Sim,\calZ\}$ are
  indistinguishable, i.e., that the output of $\calZ$ in the networks
  (\rmOne) and (\rmTwo) in \autoref{fig:comp} is statistically
  indistinguishable.

  For this, we construct a hybrid environment $\calZ_{\sigma,i}$. (It
  is depicted as the dashed box in network (\rmThree) in
  \autoref{fig:comp}.) This environment simulates the machines
  $\calZ$, $\Adv$, the protocol $\sigma$, instances
  $\pi_1,\dots,\pi_{i-1}$ of the real protocol~$\pi$, and instances
  $\Sim'_{i+1},\dots,\Sim'_n$ and
  $\rho_{i+1},\dots,\rho_n$ of the dummy-simulator $\Sim'$ and the ideal
  protocol~$\rho$, respectively. The communication between $\calZ$,
  $\Adv$, and $\sigma$ is directly forwarded by $\calZ_{\sigma,i}$. Communication between $\Adv$
  and the $j$-th protocol instance is forwarded as follows: If
  $j<i$, the communication is simply forwarded to $\pi_j$. If
  $j>i$, the communication is routed through the corresponding
  dummy-simulator $\Sim'_j$ (which is then supposed to
  transform these messages into a form suitable for $\rho_i$). And
  finally, if $j=i$, the communication is passed to the adversary/simulator outside of
  $\calZ_{\sigma,i}$. Communication between $\sigma$ and the instances
  of $\pi$ or $\rho$ is directly forwarded.

We will now show that there is a negligible function $\mu$ such that
$\abs{\Pr[\Exec_{\pi\cup\{\Advd,\calZ_{\sigma,i}\}}(k,z)=1]-\Pr[\Exec_{\rho\cup\{\Sim',\calZ_{\sigma,i}\}}(k,z)=1]}\leq\mu(k)$
for any security parameter $k$ and any $i=1,\dots,n$. For this,
construct an environment $\calZ_\sigma$ which expects as its initial
input a pair $(i,z)$, and then runs $\calZ_{\sigma,i}$ with input
$z$. Since $\pi\cup\{\Advd,\calZ\}$ and $\rho\cup\{\Sim',\calZ\}$ are
indistinguishable for all quantum-polynomial-time environments $\calZ$,  there exists a negligible function $\mu$ such
that the difference of
$
\Pr[\Exec_{\pi\cup\{\Advd,\calZ_{\sigma,i}\}}(k,z)=1] =
\Pr[\Exec_{\pi\cup\{\Advd,\calZ_{\sigma}\}}(k,(i,z))=1]$
and
$\Pr[\Exec_{\rho\cup\{\Sim',\calZ_{\sigma,i}\}}(k,z)=1]
=\Pr[\Exec_{\rho\cup\{\Sim',\calZ_{\sigma}\}}(k,(i,z))=1]
$
is bounded by $\mu(k)$ for all~$i,k,z$.

  The game $\Exec_{\pi\cup\{\Advd,\calZ_{\sigma,i}\}}(k,z)$ is depicted
  as network (\rmThree) in \autoref{fig:comp} (except that we denoted
  the external copy of $\pi$ with $\pi_i$). Observe that
  $\Exec_{\rho\cup\{\Sim',\calZ_{\sigma,i+1}\}}(k,z)$ (note the changed
  index $i+1$) contains the same machines as  $\Exec_{\pi\cup\{\Advd,\calZ_{\sigma,i}\}}(k,z)$ (when unfolding the
  simulation performed by $\calZ_{\sigma,i}$ into individual machines)
  up to the fact that the communication with the $i$-th instance of $\pi$
  is routed through the dummy-adversary $\Advd$. However, the
  latter just forwards messages, so
  $\pi\cup\{\Advd,\calZ_{\sigma,i}\}$ and
  $\rho\cup\{\Sim',\calZ_{\sigma,i+1}\}$ are perfectly indistinguishable.

  Using the triangle inequality, it
  follows that 
  $\abs{\Pr[\Exec_{\pi\cup\{\Advd,\calZ_{\sigma,n}\}}(k,z)=1]
    -
    \Pr[\Exec_{\rho\cup\{\Sim',\calZ_{\sigma,1}\}}(k,z)=1]}$ is bounded by
  $n\cdot\mu(k)$ which is negligible.  Moreover,
  $\Exec_{\pi\cup\{\Advd,\calZ_{\sigma,n}\}}(k,z)$ and
  $\Exec_{\sigma^\pi\cup\{\Adv,\calZ\}}(k,z)$ describe the same game (up to
  unfolding of simulated submachines and up to one instance of the
  dummy-adversary). Similarly,
  $\Exec_{\rho\cup\{\Sim',\calZ_{\sigma,1}\}}(k,z)$ and
  $\Exec_{\sigma^\rho\cup\{\Sim,\calZ\}}(k,z)$ describe the same game (up to
  unfolding of simulated submachines).  Thus
  $\babs{\Pr[\Exec_{\sigma^\pi\cup\{\Adv,\calZ\}}(k,z)=1]-
  \Pr[\Exec_{\sigma^\rho\cup\{\Sim,\calZ\}}(k,z)=1]}$ is negligible and thus
  $\sigma^\pi\cup\{\Adv,\calZ\}$ and $\sigma^\rho\cup\{\Sim,\calZ\}$
  are indistinguishable. Furthermore, since $\Adv$ and $\Sim'$ are
  quantum-polynomial-time, so is $\Sim$.

  Since this holds for all $\calZ$, and the construction of $\Sim$
  does not depend on $\calZ$, we have that~$\sigma^\pi$
  computationally quantum-UC-emulates~$\sigma^\rho$.

  The case of statistical quantum-UC-security is shown analogously,
  except that $\Adv$ and $\calZ$ may be unbounded, and $\Sim$ is only
  quantum-polynomial-time if $\Adv$ is.
  \qed
\end{proof}
}
\fullonly{\usedelayedtext{proof theo:comp}}

\begin{shortversion}
  The proof of \autoref{theo:comp} is very similar to that given in
  \cite{Canetti:2001:Security} and given in Appendix~\ref{app:proofs:sec:quc}.
\end{shortversion}

\section{Relating classical and quantum-UC}

We call a machine classical\index{classical!machine}\index{machine!classical} if its state transition operator is
classical. A protocol\index{protocol!classical}\index{classical!protocol} is classical if all its machines are classical.

Using this definition we can reformulate the definition of statistical
classical
UC in our framework.

\begin{definition}[Statistical classical-UC-security]\label{def:stat.class.uc}
  Let protocols $\pi$ and $\rho$ be given. We say $\pi$
  \emph{statistically classical-UC-emulates
    $\rho$}\index{UC!classical}\index{classical!UC} iff for every set
  $C\subseteq \parties_\pi$ and for every classical adversary $\Adv$
  there is a classical simulator $\Sim$ such that for every classical
  environment $\calZ$, $\pi^C\cup\{\Adv,\calZ\}$ and
  $\rho^C\cup\{\Sim,\calZ\}$ are indistinguishable.  We furthermore
  require that if $\Adv$ is probabilistic-polynomial-time, so is $\Sim$.
\end{definition}

\noindent Note that classical statistical UC is essentially the same as the
notion of statistical UC-security defined by Canetti
\cite{Canetti:2001:Security}.\footnote{Details such as the machine
  model and message scheduling are defined differently, of course. But
  since these details also considerably change between different
  versions of the full version \cite{Canetti:2005:Security:Full}, we
  feel justified in saying that the notion of statistical classical UC
  is essentially the same as that formulated by Canetti.} Thus, known
results for statistical UC-security carry over to the setting of
\autoref{def:stat.class.uc}.

The next theorem guarantees that if a classical protocol is
statistically classical UC-secure, then it is also statistically
quantum-UC-secure. This allows, e.g., to first prove the security of a
protocol in the (usually much simpler) classical setting, and then to
compose it with quantum protocols using the universal composition
theorem (\autoref{theo:comp}).

\begin{theorem}[Quantum lifting
  theorem]\label{theo:lift.stat}\index{quantum lifting}\index{lifting!quantum}
  Let $\pi$ and $\rho$ be classical protocols.  Assume that $\pi$
  statistically classical-UC-emulates $\rho$. Then $\pi$ statistically
  quantum-UC-emulates~$\rho$.
\end{theorem}

\begin{proof}
  Given a machine $M$, let $\calC(M)$ denote the machine which behaves
  like $M$, but measures incoming messages in the computational basis
  before processing them, and measures outgoing messages in the
  computational basis. More precisely, the superoperator
  $\calE_{\calC(M)}^{(k)}$ first invokes $\Eclass$
  on $\Hclass\otimes\Hquant$, then invokes $\calE^{(k)}_M$ on
  $\Hstate\otimes\Hclass\otimes\Hquant$, and then again invokes
  $\calE_{class}$ on $\Hclass\otimes\Hquant$.  Since it is possible to
  simulate quantum Turing machines on classical Turing machines (with
  an exponential overhead), for every machine $M$, there exists a
  classical machine $M'$ such that $\calC(M)$ and $M'$ are perfectly
  indistinguishable.\footnote{More precisely, for any set of machines
    $N$, the networks $N\cup\{M\}$ and $N\cup\{\calC(M)\}$ are
    perfectly indistinguishable.}

  We define the classical dummy-adversary
  $\Advdc$\index{dummy-adversary!classical}\index{adversary!classical
    dummy}\index{classical dummy adversary} to be the classical machine
  that is defined like $\Advd$ (\autoref{def:adv.dummy}), except that
  in each invocation, it first measures $\Hclass$, $\Hquant$, and
  $\Hstate$ in the computational basis (i.e., it applies $\Eclass$ to
  $\Hstate\otimes\Hclass\otimes\Hquant$) and then proceeds as does
  $\Advd$. Note that $\Advdc$ is probabilistic-polynomial-time.

  By \autoref{lemma:dummy.complete}, we only need to show that for any set $C$
  of corrupted parties, there exists a quantum-polynomial-time machine
  $\Sim$ such that for every machine $\calZ$ the real model
  $\pi^C\cup\{\calZ,\Adv_\mathit{dummy}\}$ and the ideal model
  $\rho^C\cup\{\calZ,\Sim\}$ are indistinguishable. 

  The protocol $\pi$ is classical, thus $\pi^C$ is classical, too, and thus all messages forwarded
  by $\Advd$ from $\pi^C$ to $\calZ$ have been measured in the
  computational basis by $\pi^C$, and all messages forwarded by
  $\Advd$ from $\calZ$ to $\pi^C$ will be measured by $\pi^C$ before
  being used. Thus, if $\Adv$ would additionally measure all messages
  it forwards in the computational basis, the view of $\calZ$ would
  not be modified. More formally,
  $\pi^C\cup\{\calZ,\Adv_\mathit{dummy}\}$ and
  $\pi^C\cup\{\calZ,\Advdc\}$ are perfectly indistinguishable.  Furthermore,
  since both $\pi^C$ and $\Advdc$ measure all messages upon sending
  and receiving, $\pi^C\cup\{\calZ,\Advdc\}$ and
  $\pi^C\cup\{\calC(\calZ),\Advdc\}$ are perfectly indistinguishable. 
  Since it is possible to simulate quantum machines on classical
  machines (with an exponential overhead), there exists a classical
  machine $\calZ'$ that is perfectly indistinguishable from
  $\calC(\calZ')$. Then
  $\pi^C\cup\{\calC(\calZ),\Advdc\}$ and $\pi^C\cup\{\calZ',\Advdc\}$
  are perfectly indistinguishable.  Since $\Advdc$ and $\calZ'$ are classical and
  $\Advdc$  is polynomial-time, there exists a classical probabilistic-polynomial-time simulator $\Sim$
  (whose construction is independent of $\calZ'$) such that
  $\pi^C\cup\{\calZ',\Advdc\}$ and $\rho^C\cup\{\calZ',\Sim\}$ are indistinguishable.

  Then $\rho^C\cup\{\calZ',\Sim\}$ and
  $\rho^C\cup\{\calC(\calZ),\Sim\}$ are perfectly indistinguishable by
  construction of $\calZ'$. And since both $\rho^C$ and $\Sim$
  measure all messages they send and receive,
  $\rho^C\cup\{\calC(\calZ),\Sim\}$ and $\rho^C\cup\{\calZ,\Sim\}$
  are perfectly indistinguishable.  

  Summarizing, we have that $\pi^C\cup\{\calZ,\Adv_\mathit{dummy}\}$
  and $\rho^C\cup\{\calZ,\Sim\}$ are indistinguishable for all
  quantum-polynomial-time environments $\calZ$. Furthermore, $\Sim$ is classical
  probabilistic-polynomial-time and hence quantum-polynomial-time and its construction does not
  depend on the choice of $\calZ$. Thus $\pi$ statistically quantum-UC-emulates~$\rho$. \qed
\end{proof}

\begin{fullversion}
\subsection{The computational case}

We now formulate a computational analogue to the quantum lifting
theorem (\autoref{theo:lift.stat}) from the previous section.  We
cannot, however, expect a theorem of the following form: If $\pi$
computationally classical-UC-emulates $\rho$, then $\pi$
computationally quantum-UC-emulates $\rho$. For example, if the
security of $\pi$ is based on the hardness of the discrete logarithm,
then $\pi$ may computationally classical-UC-emulate $\rho$, but
certainly $\pi$ does not computationally quantum-UC-emulate $\rho$ --
a quantum-polynomial-time adversary can easily compute discrete
logarithms using Shor's algorithm \cite{Shor:1994:Algorithms}.  Thus,
in order to get a computational quantum lifting theorem, we need to
give the adversary in the classical setting the same computational
power as in the quantum setting. Classical machines that are as
powerful as quantum-polynomial-time machines, we call QPPT machines.

\begin{definition}[Quantum-strong
  PPT]\index{QPPT}\index{quantum-strong PPT|see{QPPT}}\index{PPT!quantum-strong|see{QPPT}}
  A classical machine $M$ is said to be QPPT (quantum-strong
  probabilistic polynomial-time) if there is a quantum-polynomial-time
  machine $\Tilde M$ such that for any network $N$, $N\cup\{M\}$ and
  $N\cup\{\Tilde M\}$ are perfectly indistinguishable (short: $M$ and
  $\Tilde M$ are perfectly indistinguishable).
\end{definition}

\begin{definition}[QPPT classical UC
  security]\index{UC!QPPT classical}\index{QPPT classical
    UC}\index{classical UC!QPPT}
  Let protocols $\pi$ and $\rho$ be given. We say $\pi$ \emph{QPPT
    classical-UC-emulates $\rho$} iff for every set
  $C\subseteq \parties_\pi$ and for every QPPT adversary $\Adv$ there
  is a QPPT simulator $\Sim$ such that for every QPPT environment~$\calZ$, the networks $\pi^C\cup\{\Adv,\calZ\}$ and
  $\rho^C\cup\{\Sim,\calZ\}$ are indistinguishable.
\end{definition}

\begin{theorem}[Quantum lifting theorem --
  computational]\label{theo:lift.comp}\index{quantum lifting!computational}
  Let $\pi$ and $\rho$ be classical protocols.  Assume that $\pi$
   QPPT classical-UC-emulates $\rho$. Then $\pi$ 
  computationally quantum-UC-emulates~$\rho$.
\end{theorem}

\begin{proof}
  We define $\calC(M)$ and $\Advdc$ as in the proof of \autoref{theo:lift.stat}.

  By \autoref{lemma:dummy.complete}, we only need to show that for any set $C$ of corrupted parties,
  there exists a quantum polynomial-time machine $\Sim$ such that for
  every quantum-polynomial-time machine $\calZ$ the real model
  $\pi^C\cup\{\calZ,\Adv_\mathit{dummy}\}$ and the ideal model
  $\rho^C\cup\{\calZ,\Sim\}$ are indistinguishable.

  The protocol $\pi$ is classical, so is $\pi^C$ is classical, and thus all messages forwarded
  by $\Advd$ from $\pi^C$ to $\calZ$ have been measured in the
  computational basis by $\pi^C$, and all messages forwarded by
  $\Advd$ from $\calZ$ to $\pi^C$ will be measured by $\pi^C$ before
  being used. Thus, if $\Adv$ would additionally measure all messages
  it forwards in the computational basis, the view of $\calZ$ would
  not be modified. More formally,
  $\pi^C\cup\{\calZ,\Adv_\mathit{dummy}\}$ and
  $\pi^C\cup\{\calZ,\Advdc\}$ are perfectly indistinguishable.  Furthermore,
  since both $\pi^C$ and $\Advdc$ measure all messages upon sending
  and receiving, $\pi^C\cup\{\calZ,\Advdc\}$ and
  $\pi^C\cup\{\calC(\calZ),\Advdc\}$ are indistinguishable. By
  definition of QPPT machines, and since $\calC(\calZ)$ is
  quantum-polynomial-time, there is a QPPT machine $\calZ'$ that is
  perfectly indistinguishable from $\calC(\calZ)$. Then
  $\pi^C\cup\{\calC(\calZ),\Advdc\}$ and $\pi^C\cup\{\calZ',\Advdc\}$
  are perfectly indistinguishable.  Since $\Advdc$ and $\calZ'$ are QPPT
  machines, there exists a QPPT simulator $\Sim'$
  (whose construction is independent of $\calZ'$) such that
  $\pi^C\cup\{\calZ',\Advdc\}$ and $\rho^C\cup\{\calZ',\Sim'\}$ are indistinguishable.

  Then $\rho^C\cup\{\calZ',\Sim'\}$ and
  $\rho^C\cup\{\calC(\calZ),\Sim'\}$ are perfectly indistinguishable by
  construction of $\calZ'$. And since both $\rho^C$ and $\Sim'$
  measure all message they send and receive,
  $\rho^C\cup\{\calC(\calZ),\Sim'\}$ and $\rho^C\cup\{\calZ,\Sim'\}$
  are perfectly indistinguishable. Since $\Sim'$ is a QPPT machine,
  by definition there exists a quantum-polynomial-time machine $\Sim$ such that
  $\Sim$ and $\Sim'$ are perfectly indistinguishable.  Then $\rho^C\cup\{\calZ,\Sim'\}$ and
  $\rho^C\cup\{\calZ,\Sim\}$ are perfectly indistinguishable.

  Summarizing, we have that 
  $\pi^C\cup\{\calZ,\Adv_\mathit{dummy}\}$ and 
  $\rho^C\cup\{\calZ,\Sim\}$ are perfectly indistinguishable for all
  quantum-polynomial-time environments $\calZ$. Furthermore, $\Sim$ is
  quantum-polynomial-time and its construction does not depend on the
  choice of $\calZ$. Thus $\pi$ computationally
  quantum-UC-emulates $\rho$. \qed
\end{proof}

A word of caution: While the statistical quantum lifting theorem
(\autoref{theo:lift.stat}) can be directly applied to existing
statistically UC-secure protocols, the computational variant of this
theorem cannot be directly applied to existing proofs. Although
proving that a classical protocol is QPPT classical UC-secure is
probably simpler than directly performing the proof in the quantum
setting, at various places in a proof of QPPT classical UC-security
one has to prove that the machines one constructed from the
adversary/environment are QPPT. (This needs to be done whenever a
proof step is done by reduction, and when showing that the final
simulator is QPPT). As long as the constructed machines simulate the
original adversary as a black-box without rewinding, this will be
straightforward. However, when the constructed machine internally
rewinds a QPPT machine, showing that the constructed machine is also
QPPT will be non-trivial. Thus, to apply \autoref{theo:lift.comp} to
an existing protocol, we need to carefully revisit the original proof,
and we need to be aware of the fact that the closure properties of the
class of QPPT machines are not the same as those of the class of PPT
machines.

In this context, we formulate the following open problem: Can we
formulate the class of all QPPT machines as the class of all
probabilistic-polynomial-time machines relative to a suitable oracle?
More precisely, is the following conjecture true?

\begin{conjecture}
  There exists an oracle $\calO$ (e.g., the decision oracle of a
  BQP-complete problem) such that a classical machine $M$ is QPPT if
  and only if there exists an oracle machine $\Hat M^\calO$ which runs
  in probabilistic-polynomial-time and which is perfectly
  indistinguishable from $M$.
\end{conjecture}

A positive answer to this question would allow rewinding of QPPT
machines (since an oracle machine $\Hat M^\calO$ can be
rewound). However, the impact of such a positive answer would not be
limited to our setting; we expect that it would also allow a simple
analysis of classical protocols in the quantum stand-alone model, and
of classical zero-knowledge proofs in the quantum setting. 
\end{fullversion}

\delaytext{relation stand-alone}{
\section{Relation to the stand-alone model}
\label{sec:rel.sa}

In this section, we show that security in the quantum stand-alone
model does, in some cases, already imply quantum-UC-security. We will
need this result as a tool for reusing parts of the proof given by
Damgård et al.~\cite{Damgaard:2009:Improving} for their OT
protocol. We first review the necessary parts of the stand-alone
model as defined by Fehr and Schaffner \cite{fehr09composing}. For
details, see their paper.

The basic idea behind the stand-alone
model\index{model!stand-alone}\index{stand-alone model} is similar to that of the
UC model.  We are given a protocol $\pi$ and a functionality $\calF$,
and we call the protocol $\pi$ secure if any attack on $\pi$ can be
simulated in an ideal model where the simulator only has access to the
functionality~$\calF$.\footnote{In the stand-alone model, one usually
  call this functionality a function because it is required to be
  non-interactive, first taking inputs from all parties, and then
  sending the computed outputs to all parties.}  We will only need the
special case of a two-party protocol in which Alice does not take any
input. In this case, we say that the protocol $\pi$ implements $\calF$
in the statistical quantum stand-alone model for corrupted Bob if the
following holds: For any adversary $\Adv$, there is a simulator $\Sim$
such that such that for any quantum state $\rho_\mathit{adv}$, the
trace distance between the states $\rho_\mathit{real}$ and
$\rho_\mathit{ideal}$ is negligible. Here, the state $\rho_\mathit{real}$ is
defined to be the joint state consisting of the output of Alice and of the
adversary after a protocol execution in which the adversary gets
$\rho_\mathit{adv}$ as his initial input.  The state
$\rho_\mathit{ideal}$ is defined to be the joint state consisting of the
output of Alice and of the simulator after an execution in which the
simulator first gets $\rho_\mathit{adv}$ as his initial input, then
may give arbitrary inputs to $\calF$ in the name of Bob, then gets the
outputs for Bob from $\calF$, and then produces his output.

\begin{theorem}\label{theo:sa.vs.uc}
  Fix a protocol $\pi$ with parties Alice and Bob (not using any ideal
  functionality).  Assume that in
  this protocol Alice takes no input, and that Alice does not accept
  messages after sending her output.

  Assume that the protocol $\pi$ implements a two-party functionality
  $\calF$ in the statistical quantum stand-alone model for corrupted
  Bob.

  Assume that the corresponding simulator is quantum-polynomial-time,
  that the simulator internally simulates the adversary as a black-box
  (and in particular, the description of the simulator does not
  otherwise depend on the adversary), that the simulator does not
  rewind the adversary, and that the simulator outputs the state
  output by the internally simulated adversary.

  Then $\pi$ statistically quantum-UC-emulates $\calF$ in the case of corrupted Bob.
\end{theorem}

\begin{figure}
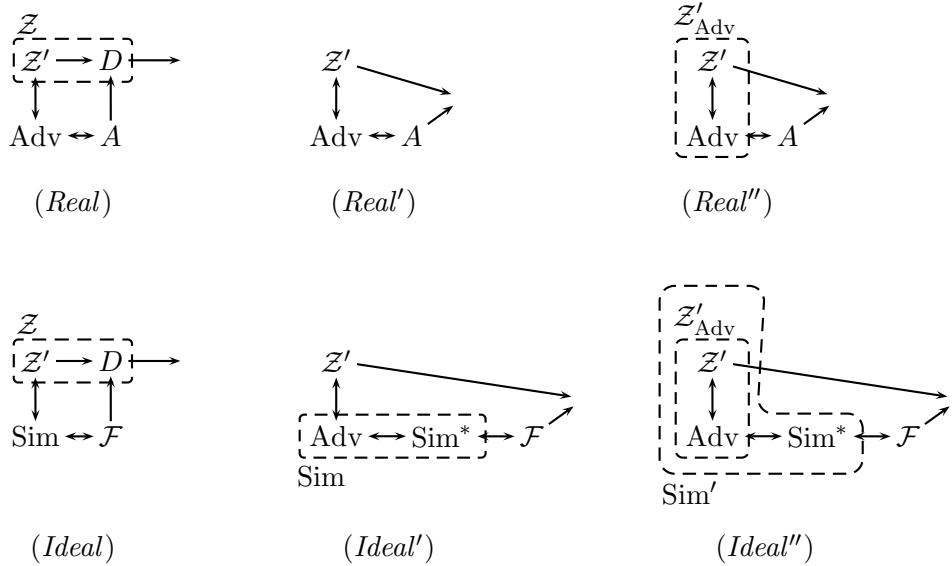

  \centering{\figureStandalone}
  \caption{\label{fig:sa.vs.uc}Networks occurring in the proof of
    \autoref{theo:sa.vs.uc}. Dashed boxes represent machines that
    internally simulate other machines. Arrows between machines
    represent communication, and arrows leaving the network represent
    the overall output of the network (indistinguishability is defined
    in terms of that output). Dummy-parties and corruption parties are
    omitted for simplicity.
  }
\end{figure}

\begin{proof}
  Fix an environment $\calZ$. By \autoref{lemma:dummy.complete}, we have to
  construct a simulator $\Sim$ such that the probability that $\calZ$
  outputs $1$ in the real and ideal model is negligibly close. (This
  simulator needs to be independent of the choice of $\calZ$.) Here, the
  real model $\mathit{Real}$ consists of the environment $\calZ$, the dummy-adversary
  $\Adv:=\Advd$, the honest party $A$ (Alice), and the corruption party $B^C$.  The
  ideal model $\mathit{Ideal}$ consists of the environment $\calZ$, the simulator
  $\Sim$, the functionality $\calF$, the dummy-party $\Tilde A$, and
  the corruption party~$B^C$. 

  Alice does not accept any messages after sending her output, 
  so we can assume without loss of generality that $\calZ$ does not
  send any messages to Alice after receiving her output. Since $\Adv$
  is the dummy-adversary, we can assume that $\calZ$ also does not
  send any messages to $\Adv$ after receiving Alice's output (since
  these messages would be routed through $\Adv$ through $B^C$ and then
  to Alice and ignored).
  Thus, we can assume without loss of generality that after receiving
  Alice's output, $\calZ$ does not send any messages, but performs a
  measurement $D$ on its state and Alice's output with some
  outcome $d\in\bit$. Then $\calZ$ terminates with output $d$.  Thus
  we can represent $\calZ$ as consisting internally of two machines
  $\calZ'$ and $D$. The machine $D$ gets the outputs of $\calZ'$ and
  Alice and outputs $d$.  This situation is depicted in
  \autoref{fig:sa.vs.uc}, network $\mathit{Real}$.

  We then define a network $\mathit{Real'}$ which contains $\calZ'$
  instead of $\calZ$. See \autoref{fig:sa.vs.uc}. Let
  $\rho(\mathit{Real'})$ denote the joint output of $\calZ'$ and
  Alice. Note that this output is not a single bit (as in
  \autoref{def:stat.quc}) but a quantum state. Note that when applying the
  measurement $D$ to $\rho(\mathit{Real'})$, the distribution of the
  measurement outcome is the distribution of the output of $\calZ'$ in
  $\mathit{Real}$.

  We then define a network $\mathit{Real''}$ which results from
  $\mathit{Real'}$ by replacing $\calZ'$ and $\Adv$ by a single
  machine $\calZ_\Adv$ which internally simulates $\calZ'$ and
  $\Adv$. See \autoref{fig:sa.vs.uc}.
  Then $\rho(\mathit{Real'})=\rho(\mathit{Real''})$.
  
  Now, since $\pi$ implements $\calF$ in the statistical quantum stand-alone model, and
  since $\calZ'_\Adv$ is a valid adversary in the quantum stand-alone
  model (it only interacts with the honest parties, but does not
  provide inputs or get the outputs), we have that there is a
  simulator $\Sim'$ such that the trace distance between
  $\rho(\mathit{Real''})$ and $\rho(\mathit{Ideal''})$ is
  negligible. Here $\mathit{Ideal''}$ is the network consisting of
  $\Sim'$ and $\calF$.

  By assumption, the simulator $\Sim'$ internally simulates
  $\calZ'_\Adv$ as a black box and outputs what the simulated
  $\calZ'_\Adv$ outputs. Hence we can represent $\Sim'$ as internally
  consisting of some two machines: the adversary $\calZ'_\Adv$, and
  some machine $\Sim^*$ that interacts with $\calZ'_\Adv$. 
  The construction of $\Sim^*$ does not depend on $\calZ'_\Adv$, and
  $\Sim^*$ is quantum-polynomial-time since $\Sim'$ is
  quantum-polynomial-time by assumption. The output
  of $\Sim'$ is that of $\calZ'_\Adv$. Note further that $\calZ'_\Adv$
  by construction also consists of two internally simulated machines
  $\calZ'$ and $\Adv$ and outputs what $\calZ'$ outputs. So the output
  of $\Sim'$ is that of the internal $\calZ'$. See
  \autoref{fig:sa.vs.uc}, network $\mathit{Ideal''}$.
  
  The simulator $\Sim'$ internally simulates $\calZ'$, $\Adv$, and
  $\Sim^*$. We define $\mathit{Ideal'}$ by replacing $\Sim'$ in
  $\mathit{Ideal''}$ by $\calZ'$ and $\Sim$, where $\Sim$ is defined
  to internally simulate $\Adv$ and $\Sim^*$. See
  \autoref{fig:sa.vs.uc}. Then
  $\rho(\mathit{Ideal''})=\rho(\mathit{Ideal'})$.

  Thus the trace distance between $\rho(\mathit{Real'})$ and
  $\rho(\mathit{Ideal'})$ is negligible. Furthermore, when applying
  the measurement $D$ to $\rho(\mathit{Real'})$, the distribution of
  the measurement outcome is the distribution of the output of
  $\calZ$ in $\mathit{Real}$. Similarly, when applying the
  measurement $D$ to $\rho(\mathit{Ideal'})$ is the distribution of
  the output of $\calZ$ in $\mathit{Ideal}$. Thus the statistical
  distance between the output of $\calZ'$ in $\mathit{Real}$ and in
  $\mathit{Ideal}$ is negligible. Thus $\mathit{Real}$ and
  $\mathit{Ideal}$ are indistinguishable.

  Furthermore, since $\Sim$ consists of $\Adv$ and $\Sim^*$, it is
  independent of $\calZ$. And since $\Sim^*$ is
  quantum-polynomial-time, $\Sim$ is quantum-polynomial-time if
  $\Adv$ is. Thus $\pi$ statistically quantum-UC-emulates $\calF$ in
  the case of corrupted Bob.
\qed
\end{proof}
}
\fullonly{\usedelayedtext{relation stand-alone}}

\section{Oblivious transfer}
\label{sec:ot}

\begin{definition}[OT protocols]\label{def:pi}
  \index{OT!protocol}The protocol $\piQROT$ is defined in \autoref{fig:qrot}. Fix a
  commitment scheme~$\com$. The protocol $\piQROT^\com$ is defined
  like $\piQROT$, but instead of using the functionality $\FCOM$, the
  commitment scheme $\com$ is used. The protocol $\piQOT$ is defined
  like $\piQROT$, with the following modifications: Alice takes as
  input two $\ell(k)$-bit strings $v_0,v_1$.  In
  Step~\ref{pi.send}, Alice additionally sends $t_0,t_1$ with
  $t_i:=s_i\oplus v_i$.  Bob outputs $s\oplus t_c$ instead of $s$ in
  Step~\ref{pi.output}.
\end{definition}

We first analyze $\piQROT$ and will then deduce the security of
$\piQOT$ from that of $\piQROT$.

\begin{figure}[t]
  \framebox{\parbox{\hsize}{
  \noindent\textbf{Parameters:} 
  Integers $n$, $m>n$, $\ell$, a family $\mathbf F$ of universal hash functions.

  \noindent\textbf{Parties:} The sender Alice and the recipient Bob.

  \noindent\textbf{Inputs:} Alice gets no input, Bob gets a bit $c$.

  \begin{compactenum}
  \item Alice chooses $\tilde x^A\in\bits m$ and $\tilde \theta^A\in\bases m$ and
    sends $\ket {\tilde x^A}_{\tilde \theta^A}$ to Bob. 
  \item Bob receives the state $\ket\Psi$ sent by the
    sender.  Then Bob chooses $\tilde\theta^B\in\bases m$ and measures
    the qubits of 
    $\ket\Psi$ in the bases $\tilde\theta^B$. Call the result $\tilde x^B$.
  \item For each $i$, Bob commits to $\tilde\theta^B_i$ and 
    $\tilde x^B_i$ using one instance of $\FCOM^{B\to A,1}$ each.
  \item Alice chooses a set $T\subseteq\{1,\dots,m\}$ of size $m-n$
    and sends $T$ to Bob.
  \item\label{pi.open}Bob opens the commitments of $\tilde\theta^B_i$ and $\tilde
    x^B_i$ for all $i\in T$.
  \item\label{alice.tests}Alice checks $\tilde x^A_i=\tilde x^B_i$ for all $i$ with
    $i\in T$ and $\tilde\theta^A_i=\tilde\theta^B_i$. If this test
    fails, Alice aborts.
  \item Let $x^A$ be the $n$-bit string resulting from removing the
    bits at positions $i\in T$ from~$\tilde x^A$. Define $\theta^A$,
    $x^B$, and $\theta^B$ analogously.
  \item Alice sends $\theta^A$ to Bob.
  \item Bob sets $I_c:=\{i:\theta^A_i=\theta^B_i\}$ and
    $I_{1-c}:=\{i:\theta^A_i\neq\theta^B_i\}$. Then Bob sends
    $(I_0,I_1)$ to Alice.
  \item\label{pi.send}Alice chooses $s_0,s_1\in\bits{\ell(k)}$ and
    $f_0,f_1\in\mathbf F$, output $(s_0,s_1)$, and computes
    $m_i:=s_i\oplus f_i(x^A|_{I_i})$  for
    $i=1,2$. Then Alice sends $f_0,f_1,m_0,m_1$ to Bob.
  \item\label{pi.output}Bob outputs $s:=m_c\oplus f_c(x^B|_{I_c})$.
  \end{compactenum}}}
  \caption{\label{fig:qrot}Protocol $\piQROT$ for randomized oblivious transfer.}
\end{figure}

\delaytext{ot trivial}{
\fullshort{We first state}{Finally, we need} the trivial cases (note for the uncorrupted case that
we assume secure channels):

\begin{lemma}\label{lemma:ot.uncorr}
  The protocol $\piQROT$ statistically quantum-UC-emulates
  $\FROT^{A\to B,\ell}$ in the case of no corrupted parties and in the
  case of both Alice and Bob being corrupted.
\end{lemma}
}
\fullonly{\usedelayedtext{ot trivial}}

\subsection{Corrupted Alice}

\begin{lemma}\label{lemma:ot.a}
  The protocol $\piQROT$ statistically quantum-UC-emulates $\FROT^{A\to B,\ell}$ in
  the case of corrupted Alice.
\end{lemma}

\begin{proof}
  First, we describe the structure of the real and ideal model in the
  case that the party $A$ (Alice) is corrupted:

  In the real model, we have the environment $\calZ$, the adversary
  $\Adv$, the corruption party $A^C$, the honest party $B$ (Bob), and the
  $2m$ instances of the commitment functionality $\FCOM$. The adversary controls the
  corruption party $A^C$, so effectively he controls the communication
  with Bob and the inputs of $\FCOM$. Bob's input (a choice bit $c$)
  is chosen by the environment, and the environment also gets Bob's
  output (a bitstring $s\in\bits\ell$).  See \autoref{fig:ot.a}(a).

  \begin{figure}[t]
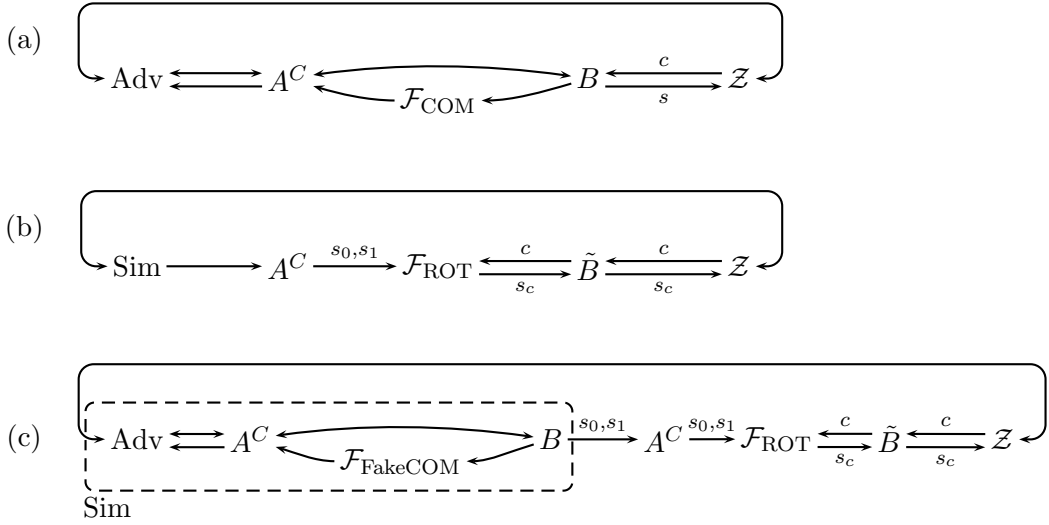

    \centering{\figureCorruptedAlice}
    \caption{\label{fig:ot.a} Networks occurring in the proof of
      \autoref{lemma:ot.a}. The dashed box represents the machine
      $\Sim$ that internally simulates $\Adv$, $A^C$, $\FFC$ and $B$.}
  \end{figure}

  In the ideal model, we have the environment $\calZ$, the simulator
  $\Sim$ (to be defined below), the corruption party $A^C$, the dummy-party $\Tilde B$, and
  the randomized OT functionality $\FROT$. The simulator $\Sim$ controls the
  corruption party $A^C$ and hence effectively chooses the inputs
  $s_0,s_1$ of $\FROT$.\footnote{Remember that, if Alice is corrupted,
    $\FROT$ behaves like $\FOT$ and takes inputs $s_0,s_1$ from Alice.}
  The input $c$ of $\FROT$ is chosen by the
  dummy-party $\Tilde B$ and thus effectively by the environment
  $\calZ$. The output $s:=s_c$ of $\FROT$ is given to the dummy-party
  $\Tilde B$ and thus effectively to the environment $\calZ$. See
  \autoref{fig:ot.a}(b).

  To show \autoref{lemma:ot.a}, we need to find a simulator $\Sim$
  such that, for any environment~$\calZ$, the real model and the ideal
  model are indistinguishable. To do so, we start with the real model,
  and change the machines in the real model step-by-step until we end
  up with the ideal model containing a suitable simulator $\Sim$
  (which we define below in the description of
  \autoref{step:make.sim}). In each step, we show that network 
  before and after the step are perfectly indistinguishable.

  \nextstep{step:equivocal.com} We replace $\FCOM$ by a commitment
  functionality $\FFC$ in which Bob (the sender) can cheat. That is,
  in the commit phase, $\FFC$ expects a message $\mathtt{commit}$ from
  $B$ (instead of $(\mathtt{commit},x)$), and in the open phase, $\FFC$ expects a
  message $(\mathtt{open},x)$ (instead of $\mathtt{open}$) and then
  sends $(\mathtt{open},x)$ to Alice.  We also change Bob's implementation
  accordingly, i.e., when Bob should commit to a bit $b$, he stores
  that bit $b$ and gives it to $\FFC$ when opening the commitment.
  Obviously, this change leads to a perfectly indistinguishable network
  (since Bob still opens the commitment in the same way).

  \nextstep{step:measure.late} Since Bob uses $\FFC$ instead of
  $\FCOM$, he does not use the outcomes~$\tilde x_i^B$ of his
  measurements before Step~\ref{pi.open} (for $i\in T$) or
  Step~\ref{pi.output} (for $i\notin T$) of the protocol. Thus, we modify Bob so that
  he performs the measurements with outcomes $\tilde x_i^B$ ($i\in T$)
  in  Step~\ref{pi.open} (in particular, after learning $T$), and the
  measurements with outcomes $x_i^B$ in 
  Step~\ref{pi.output}. Delaying the measurements leads to a perfectly
  indistinguishable network.

  \nextstep{step:alice.base} The bits $x^B_i$ with $i\in I_{1-c}$ are never
  used by Bob.  Thus
  we can modify Bob to use the bases $\theta^A_i$ instead of
  $\theta^B_i$ for these bits without changing the output of
  $\calZ$. Furthermore, since $\theta^A_i=\theta^B_i$ for $i\in I_c$,
  we can modify Bob to also use the bases $\theta^A_i$ instead of
  $\theta^B_i$ when measuring $x^B_i$ with $i\in I_c$. Summarizing, we
  modify Bob to use $\theta^A$ instead of $\theta^B$, and we get a
  perfectly indistinguishable network.

  \nextstep{step:i0i1.random} The bases $\theta^B$ are chosen randomly
  by Bob, and they are only used to compute the sets $I_0$ and
  $I_1$. We change Bob to instead pick $(I_0,I_1)$ as a random
  partition of $\{1,\dots,n\}$. Since this leads to the same
  distribution of $(I_0,I_1)$ and since $\theta^B$ is not used
  elsewhere, this leads to a perfectly indistinguishable network.

  \nextstep{step:compute.s0s1} In Step~\ref{pi.output}, we change Bob
  to compute $s_i:=m_i\oplus f_i(x^B|_{I_i})$ for $i=0,1$ and to
  output $s:=s_c$. This leads to the same value of $s$ as the original
  computation $s:=m_c\oplus f_c(x^B|_{I_c})$, hence the resulting
  network is perfectly indistinguishable from the previous one.
  Note that now, Bob only uses the choice bit $c$ to pick which of the
  two values $s_0,s_1$ to output.

  \nextstep{step:make.sim} We now construct a machine $\Sim$ that
  internally simulates the machines $\Adv$, $A^C$, $\FFC$, and Bob. We let $\Sim$
  run with an (external) corruption party $A^C$, and when (the
  simulated) Bob computes $s_0,s_1$ in Step~\ref{pi.output}, $\Sim$
  instructs the (external) corruption party $A^C$ to input $s_0,s_1$
  into $\FROT$ (instead of letting Bob output $s=s_c$). Then $\FROT$
  will, given input $c$ from the dummy-party $\TIlde B$, output $s_c$
  to the dummy-party $\TIlde B$. The dummy-party $\Tilde B$ then
  forwards $s_c$ to the environment $\calZ$. See \autoref{fig:ot.a}(c). The only difference with
  respect to the previous network (besides a regrouping of machines)
  is that now $s_c$ is computed by $\FROT$ from $s_0,s_1$. However,
  $\FROT$ computes $s_c$ in the same way as Bob would have done. Thus,
  the resulting network is perfectly indistinguishable from the
  previous one.

  Since the network from \autoref{step:make.sim}
  (\autoref{fig:ot.a}(c)) is identical to the
  ideal model (\autoref{fig:ot.a}(b)), and since the real model is perfectly
  indistinguishable from the network from \autoref{step:make.sim}, we
  have that the real and the ideal network are perfectly indistinguishable.

  Furthermore, $\Sim$ is quantum-polynomial-time if $\Adv$ is, and the
  construction of $\Sim$ does not depend on the choice of the
  environment $\calZ$.  Thus the protocol $\piQROT$ statistically quantum-UC-emulates
  $\FROT^{A\to B,\ell}$ in the case of corrupted Alice.
  \qed
\end{proof}

\delaytext{corrupted bob}{
\fullshort{\subsection{Corrupted Bob}}{\section{Security of $\piQROT$
    for corrupted Bob}}
\label{sec:corrupted-bob}

We call a commitment scheme trivially extractable if, given the
messages exchanged during the commit phase, it is efficiently possible
to determine the value to which the commitment will be
opened. Obviously, this directly contradicts the hiding property of
the commitment, so trivially extractable commitments are not overly
useful. However, we need such commitments as an intermediate
construction in the following proofs. An example of a trivially
extractable commitment is one which sends the committed message in
clear during the commit phase.

\begin{corollary}[Stand-alone quantum OT \cite{Damgaard:2009:Improving}]\label{coro:qrot.sa}
  Let $0<\alpha<1$ and $0<\lambda<\frac14$ be constants. Assume
  $m=\lceil n/(1-\alpha)\rceil$ and $\ell=\lfloor \lambda n\rfloor$
  and that $n$ grows at least linearly in the security parameter~$k$.

  Assume that $\com$ is a statistically binding, trivially extractable
  commitment scheme. Then~$\piQROT^\com$ implements~$\FROT^{A\to
    B,\ell}$ in the
  statistical quantum stand-alone model.

  The corresponding simulator is quantum-polynomial-time, internally
  simulates the adversary as a black-box, does not rewind the
  adversary, and outputs the state output by the internally simulated
  adversary.
\end{corollary}

Note that Damgård et al.~\cite{Damgaard:2009:Improving} prove a slightly different
result. First, it only assumes that the commitment scheme $\com$ is extractable in the
so-called common reference string (CRS) model. That is, a globally
known and trusted string, the CRS, is available to all parties, and it
is possible to extract the committed value when one is allowed to
choose the CRS oneself. A trivially extractable commitment can be seen
as a special case with a zero-length CRS. Second, it only assumes that
the scheme is computationally binding, and thus only proves security
in the computational quantum stand-alone model. If we assume that the
commitment is statistically binding instead, the same proof shows
security in the statistical quantum stand-alone model. Third, they
analyze the protocol $\piQOT^\com$, but the proof trivially adapts to
$\piQROT^\com$.

\begin{lemma}\label{lemma:ot.b}
  Under the same assumptions on $n,m,\ell$ as in
  \autoref{coro:qrot.sa}, the protocol $\piQROT$ statistically quantum-UC-emulates
  $\FROT^{A\to B,\ell}$ in the case of corrupted Bob.
\end{lemma}

\begin{proof}
  Let $\com$ be the following encryption scheme: To commit to a
  message $m$, the sender sends $(\mathtt{commit},m)$, and the
  recipient always accepts the commitment. To open the commitment, the
  sender sends $\mathtt{open}$, and the recipients accepts and
  output $m$. Obviously, this commitment is not hiding. However, it is
  easily seen to be statistically binding and trivially extractable.

  Consider the protocol $\piQROT$. Here Bob sends the messages
  $(\mathtt{commit},m)$ and $\mathtt{open}$ to the commitment
  functionality, while in the protocol $\piQROT^\com$, Bob sends these
  messages directly to Alice. In other words, the machine Alice in
  $\piQROT^\com$ can be represented as a machine that internally simulates the machine
  Alice from $\piQROT$ and the ideal functionality $\FCOM$. Thus,
  as long as Alice is honest, $\piQROT^\com$ statistically quantum-UC-emulates
  $\FROT$ in the case of corrupted Bob if and only if $\piQROT$
  statistically quantum-UC-emulates $\FROT$ in the case of corrupted Bob.

  By \autoref{coro:qrot.sa}, $\piQROT^\com$ implements $\FCOM$ in the
  statistical quantum stand-alone model in the case of corrupted
  Bob with a simulator having the special properties listed in \autoref{coro:qrot.sa}.
  Thus, by \autoref{theo:sa.vs.uc}\shortonly{ in Appendix~\ref{sec:rel.sa}}, $\piQROT^\com$ statistically
  quantum-UC-emulates $\FROT$ in the case of corrupted Bob. Thus
  $\piQROT$ statistically quantum-UC-emulates $\FROT$ in the case of
  corrupted Bob.  \qed
\end{proof}
}
\fullonly{\usedelayedtext{corrupted bob}}

\begin{theorem}\label{theo:pi.rot}
  Let $0<\alpha<1$ and $0<\lambda<\frac14$ be constants. Assume
  $m=\lceil n/(1-\alpha)\rceil$ and $\ell=\lfloor \lambda n\rfloor$
  and that $n$ grows at least linearly in the security parameter.

  Then the protocol $\piQROT$ statistically quantum-UC-emulates
  $\FROT^{A\to B,\ell}$.
\end{theorem}

\delaytext{proof theo:pi.rot}{
\fullshort{\begin{proof}}{\begin{proof}[of \autoref{theo:pi.rot}]}
  Immediate from
  Lemmas~\ref{lemma:ot.uncorr}, \ref{lemma:ot.a},
  and~\ref{lemma:ot.b}.
\end{proof}
}
\fullonly{\usedelayedtext{proof theo:pi.rot}}

\begin{shortversion}
  \noindent For the case of corrupted Alice, this is shown in
  \autoref{lemma:ot.a}. The cases where both parties are honest or
  both parties are corrupted are trivial. Thus for
  \autoref{theo:pi.rot} we are left to analyze the case where Bob is
  corrupted. This case needs a considerably more involved analysis
  than the case of corrupted Alice because we have to consider the
  fact that Bob may succeed in Step~\ref{alice.tests} of $\piQROT$ but
  still have a certain amount of information about the bits
  $x^A|_{I_{1-c}}$.  A very similar analysis has already been
  performed by Damgård, Fehr, Lunemann, Salvail, and Schaffner
  \cite{Damgaard:2009:Improving} in the so-called stand-alone
  model. Fortunately, we do not need to redo their analysis; it turns
  out that -- although the stand-alone model is weaker than the
  quantum-UC-model -- the particular simulator constructed by Damgård
  et al.~is already strong enough to be used as a simulator in the
  quantum-UC-model. Thus we can reuse the result of Damgård et al.~in
  our setting and get \autoref{theo:pi.rot} without re-analyzing
  $\piQROT$.\footnote{One major difference between the UC-model and
    the stand-alone model is that in the first, the honest parties'
    inputs may depend on messages the adversary intercepts during the
    protocol run.  A simulator constructed for the stand-alone model
    usually is not able to cope with such dependencies. Thus, it turns
    out to be important that we first considered the randomized OT
    protocol $\piQROT$ and not immediately the OT protocol
    $\piQOT$. In $\piQROT$, Alice gets no input, and in particular her
    inputs may not depend on messages intercepted by the adversary.}

  The full proof of \autoref{theo:pi.rot} is given in Appendix~\ref{sec:corrupted-bob}.
\end{shortversion}

\begin{theorem}\label{theo:pi.ot}
  Let $0<\alpha<1$ and $0<\lambda<\frac14$ be constants. Assume
  $m=\lceil n/(1-\alpha)\rceil$ and $\ell=\lfloor \lambda n\rfloor$
  and that $n$ grows at least linearly in the security parameter.

  Then the protocol $\piQOT$ (\autoref{def:pi})
  statistically quantum-UC-emulates $\FOT^{A\to B,\ell}$.
\end{theorem}

\begin{proof}
  Consider the following protocol $\piQOT'$ in the $\FROT$-hybrid
  model. Given inputs $v_0,v_1\in\bits{\ell(k)}$ for Alice and a bit
  $c$ for Bob, Bob invokes $\FROT$ with input $c$. Then Alice gets
  random $s_0,s_1\in\bits{\ell(k)}$, and Bob gets $s=s_c$. Then Alice
  sends $t_0,t_1$ with $t_i:=v_i\oplus s_i$ to Bob. And Bob outputs
  $s\oplus t_c$.  It is easy to see that $\piQOT'$ statistically
  classical-UC-emulates $\FOT$.  Hence, by the quantum lifting theorem
  (\autoref{theo:lift.stat}), $\piQOT'$ statistically
  quantum-UC-emulates $\FOT$. Note that the protocol $\piQOT$ is the
  protocol resulting from replacing, in $\piQOT'$, calls to $\FROT$ by
  calls to the subprotocol $\piQROT$. Furthermore,
  $\piQROT$ statistically quantum-UC-emulates $\FROT$ by \autoref{theo:pi.rot}.
  Hence, by the composition
  theorem (\autoref{theo:comp}), $\piQOT$ statistically
  quantum-UC-emulates~$\FOT$. \qed
\end{proof}

\section{Multi-party computation}
\label{sec:multi-party}
\index{multi-party computation}

\begin{theorem}\label{theo:qmpc}
  Let $\calF$ be a classical probabilistic-polynomial-time functionality.\footnote{Subject to certain
    technical restrictions stemming from the proof by Ishai et
    al.~\cite{Ishai:2008:OT}:
    Whenever the functionality gets an input,
    the adversary is informed about the length of that input. Whenever
    the functionality makes an output, the adversary is informed about
    the length of that output and may decide when this output is to be
    scheduled.} Then there exists a protocol $\pi$ in the
  $\FCOM$-hybrid model that statistically quantum-UC-emulates
  $\calF$. (Assuming the number of protocol parties does not depend on the
  security parameter.)
\end{theorem}

\begin{proof}
  Ishai, Prabhakaran, and Sahai \cite{Ishai:2008:OT} prove the
  existence of a protocol $\rho^{\FOT}$ in the $\FOT$-hybrid model
  that statistically classical-UC-emulates $\calF$ (assuming a
  constant number of parties). By the quantum lifting theorem
  (\autoref{theo:lift.stat}), $\rho^{\FOT}$ statistically
  quantum-UC-emulates $\calF$. By \autoref{theo:pi.ot}, $\piQOT$
  statistically quantum-UC-emulates~$\FOT$. Let $\pi:=\rho^{\piQOT}$
  be the result of replacing invocations to~$\FOT$ in~$\rho^{\FOT}$ by
  invocations of the subprotocol~$\piQOT$ (as described \fullshort{in
    \autoref{sec:comp}}{before \autoref{theo:comp}}). 
  Then by the universal composition theorem
  (\autoref{theo:comp}), $\pi$ statistically
  quantum-UC-emulates~$\rho^{\FOT}$. Using the fact that
  quantum-UC-emulation is transitive
  (\autoref{lemma:ref.trans}\shortonly{ in the appendix}), it
  follows that $\pi$ statistically quantum-UC-emulates $\calF$.
  \qed
\end{proof}

\bigskip

We proceed to show that the result from  \autoref{theo:qmpc} is
possible only in the quantum setting. That is, we show that there is a
natural functionality that cannot be statistically
classical-UC-emulated in the commitment-hybrid model. \fullonly{To show this
impossibility result, we first need the following lemma.}

\delaytext{lemma:no.and}{
\begin{lemma}\label{lemma:no.and}
  There is no classical two-party protocol (that runs in a polynomial number of
  rounds) in the commitment-hybrid model that has the following properties:
  \begin{itemize}
  \item Let $a\in\bit$ denote Alice's input, and $b\in\bit$ Bob's
    input. Then Alice's and Bob's output is $a\cdot b$ with
    overwhelming probability.
  \item The view of Alice in the case $(a,b)=(0,0)$ is statistically 
    indistinguishable from the view of Alice in the case $(a,b)=(0,1)$.
  \item The view of Bob in the case $(a,b)=(0,0)$ is statistically
    indistinguishable from the view of Bob in the case $(a,b)=(1,0)$.
  \end{itemize}
  In all three cases we assume that Alice and Bob honestly follow the
  protocol (i.e., Alice and Bob are honest-but-curious). The view of a
  party consists of all messages sent and received by that party
  together with its input and random choices.
\end{lemma}

\begin{proof}
  Assume a protocol $\pi$ satisfying the properties from
  \autoref{lemma:no.and}.  We assume without loss of generality that
  the last message sent in an execution of $\pi$ contains the output
  of Alice.  We transform $\pi$ into a protocol $\pi'$ that does not
  use commitments. Namely, when Alice would commit to a value $m$, she
  instead sends $\mathtt{committed}$ to Bob, and when she would open
  that commitment, she sends $m$ to Bob. Analogously, we remove Bob's
  commitments. The resulting protocol $\pi'$ still satisfies the
  properties from \autoref{lemma:no.and} since we only consider
  honest-but-curious parties.

  We use Lemma 33 from \cite{muellerquade07long-long}: Let $U$, $\Tilde U$, $L$, $\Tilde{L}$
  be interactive machines that send only a polynomially-bounded number
  of messages. Let $\langle U,L\rangle$ denote the transcript of the
  communication in an interaction of~$U$ and~$L$.  Assume that
  $ \langle U,L\rangle \approx \langle \Tilde U,L\rangle \approx
  \langle U,\Tilde L\rangle $
  where $\approx$ denotes statistical indistinguishability.  Then
  $ \langle U,L\rangle \approx \langle \Tilde U,\Tilde L\rangle.  $

  Let $U$ be a machine executing Alice's program in $\pi'$ on input $0$,
  and let $\Tilde U$
  execute Alice's program on input $1$. Let $L$ and $\Tilde L$ execute
  Bob's program on inputs $0$ and $1$, respectively. Then the properties
  in \autoref{lemma:no.and} guarantee that
  $ \langle U,L\rangle \approx \langle \Tilde U,L\rangle \approx
  \langle U,\Tilde L\rangle $.
  Hence $ \langle U,L\rangle \approx \langle \Tilde U,\Tilde
  L\rangle$.
  This implies that the communication between Alice and Bob in $\pi'$
  is indistinguishable in the cases $a=b=0$ and $a=b=1$. This is a
  contradiction to the fact that in the first case, the last message
  contains the output $ab=0$, and in the second case, the last message
  contains the output $ab=1$.
\qed
\end{proof}
}
\fullonly{\usedelayedtext{lemma:no.and}}

\begin{definition}[AND]
  The functionality $\FAND$\index{functionality!AND}\index{AND!functionality} expects an input $a\in\bit$ from Alice and
  $b\in\bit$ from Bob. Then it sends $a\cdot b$ to Alice and Bob.
\end{definition}

\begin{theorem}[Impossibility of classical multi-party
  computation]\label{theo:no.and.uc}\index{multi-party
    computation!classical impossibility}
  There is no classical probabilistic-polynomial-time protocol $\pi$ in the $\FCOM$-hybrid model such that
  $\pi$ statistically classical-UC-emulates $\FAND$.
\end{theorem}

\delaytext{proof theo:no.and.uc}{
\fullshort{\begin{proof}}{\begin{proof}[of \autoref{theo:no.and.uc}]}
  The statistical UC-security of $\pi$ would imply the properties
  listed in \autoref{lemma:no.and}. Hence by \autoref{lemma:no.and}
  such a protocol $\pi$ does not exist.\qed
\end{proof}
}
\fullonly{\usedelayedtext{proof theo:no.and.uc}}

\begin{shortversion}
  \noindent The proof is given in Appendix~\ref{app:proofs:sec:multi-party}.
\end{shortversion}

\begin{fullversion}
  \section{Conclusions}
  \label{sec:concl}

  We have given a definition of quantum-UC-security that provides strong
  composability guarantees for quantum protocols. We have shown that
  in this model, it is possible to construct statistically secure
  oblivious transfer protocols given only commitments. Furthermore, we
  showed that a protocol which is secure in the statistical
  \emph{classical} UC model is also secure in the statistical
  \emph{quantum} UC model. This simplifies the modular design of
  quantum protocols and allows us to construct UC-secure general multi-party
  computation protocols given only commitments.

  Directions for future work include:
  \begin{itemize}
  \item Combine the UC framework and the bounded quantum storage
    model. In this model, Damgård, Fehr, Salvail, and Schaffner
    \cite{Damgaard:2005:BoundedQuantum} have constructed statistically
    hiding and binding commitment schemes and statistically secure OT
    protocols. If variants of these protocols can be shown secure in
    the UC framework, this would allow to construct general UC-secure
    multi-party computation protocols, only assuming that the
    adversary has a certain upper bound on his quantum storage.
  \item Combine our result with the protocols for long-term classical
    UC-secure commitments by Müller-Quade and Unruh
    \cite{muellerquade07long} (see \autoref{sec:howto}). If their
    protocols can be shown to be secure in the quantum setting, this
    would enable general long-term secure multi-party computation
    based on practical setup-assumptions (the availability of
    signature cards).
  \item Find efficient constructions. Our protocol invokes a
    commitment for each qubit sent by Alice. In some settings, a
    commitment can be quite expensive. For example, commitment
    protocols in the bounded quantum storage model have a large
    quantum communication complexity. In this setting, the efficiency
    of our protocol could be improved considerably if we were able to
    use few string commitments instead of committing to each bit
    individually.
  \item Find analogues to the quantum lifting theorem in other
    security models. In the stand-alone model, it is an open question
    whether classically secure protocols are secure in the quantum
    setting, too. Similarly, we do not know whether classically secure
    zero-knowledge proofs are in general secure against quantum
    adversaries.
  \end{itemize}

\paragraph{Acknowledgements.} I thank Jörn Müller-Quade for the
original inspiration for this work and Christian
Schaffner for valuable discussions.
\end{fullversion}

\begin{shortversion}
  \appendix

  \newpage

  \usedelayedtext{compositionality restrictions}

  \section{Deferred proofs for Section~\ref{sec:quc}}
  \label{app:proofs:sec:quc}

  \usedelayedtext{lemma:ref.trans}

  \usedelayedtext{proof lemma:dummy.complete}

  \usedelayedtext{proof theo:comp}

  \usedelayedtext{relation stand-alone}

  \usedelayedtext{corrupted bob}

  \usedelayedtext{ot trivial}

  \noindent  We now have everything needed to show the security of $\piQROT$:

  \usedelayedtext{proof theo:pi.rot}

  \section{Deferred proofs for Section~\ref{sec:multi-party}}
  \label{app:proofs:sec:multi-party}

  \usedelayedtext{lemma:no.and}

  \usedelayedtext{proof theo:no.and.uc}

\end{shortversion}

\fullshort{\bibliographystyle{alpha}}{\bibliographystyle{abbrv}}
\bibliography{quantum-uc-ot}

\fullonly\printindex

\end{document}